%% file: main.tex
\documentclass{llncs}

\input{styles.tex}
\input{macro.tex}

\begin{document}

\title{Renegotiation and recursion in Bitcoin contracts}

\author{Massimo Bartoletti\inst{1} \and
Maurizio Murgia\inst{2} \and
Roberto Zunino\inst{2}}
\institute{University of Cagliari, Italy \and
University of Trento, Italy}

\maketitle

\input{abstract.tex}

\input{intro.tex}

\input{bitml.tex}

\input{compiler.tex}

\input{example.tex}

\input{discussion.tex}

\input{conclusions.tex}
  \input{ack.tex}

\bibliographystyle{splncs04}
\bibliography{main}

\appendix

\newpage
\section{Supplementary material}
\label{app}

\input{semantics-rules.tex}

\input{computational-soundness.tex}

\newpage
\setcounter{tocdepth}{1}
\listoffixmes

\end{document}

%% file: styles.tex
\usepackage[utf8]{inputenc}
\usepackage[dvipsnames]{xcolor}

\usepackage{centernot}

\usepackage{siunitx}
\usepackage{textcomp}
\usepackage{caption}

\usepackage{pst-func}
\usepackage{amssymb}
\usepackage{graphicx}

\usepackage{paralist}
\usepackage{environ}
\usepackage{pgfplots}
\usepackage{tikz}
\usetikzlibrary{shadows}

\usepackage{mathtools}
\usepackage{array,multirow}

\usepackage{url}

\usepackage{pgfplotstable}
\usepackage{threeparttable}

\usepackage{etoolbox}
\usepackage{arydshln}

\usepackage{xspace} % Smart spaces after commands
 
\usepackage{dirtytalk} % for \say command
\usepackage{stmaryrd} % for llbracket and rrbracket\textbf{\textbf{}}

\usepackage[inline,shortlabels]{enumitem} 
\newlist{inlinelist}{enumerate*}{1}
\setlist*[inlinelist,1]{%
  label=(\roman*),
}

\usepackage{xifthen}  % Extended conditional commands
\usepackage{ifthen}

\usepackage{nicefrac}

\usepackage{hyperref}
\usepackage{cleveref}
\usepackage{bigints}

\usetikzlibrary{pgfplots.dateplot}
\usetikzlibrary{matrix,chains,positioning,decorations.pathreplacing}
\usetikzlibrary{patterns,calc,fit,arrows}
\usetikzlibrary{shapes.geometric}

\usepackage[final,nomargin,inline,index]{fixme} % Simplified management of FIXME's
\fxusetheme{color}

\FXRegisterAuthor{bart}{anbart}{\color{magenta} {\underline{bart}}}
\FXRegisterAuthor{tizi}{antizi}{\color{red} {\underline{tizi}}}
\FXRegisterAuthor{ste}{anste}{\color{blue} {\underline{ste}}}
\FXRegisterAuthor{nico}{annico}{\color{ForestGreen} {\underline{nico}}}
\FXRegisterAuthor{zun}{anzun}{\color{purple} {\underline{zun}}}
\FXRegisterAuthor{MM}{anMM}{\color{blue} {\underline{maurizio}}}

\hypersetup{
  breaklinks   = true,
  colorlinks   = true, %Colours links instead of ugly boxes
  urlcolor     = blue, %Colour for external hyperlinks
  linkcolor    = blue, %Colour of internal links
  citecolor    = red   %Colour of citations
}
\hypersetup{final} % Needed to generate hyperlinks even in draft mode (ERROR IN MAC)

\usepackage{listings}
\definecolor{listingBG}{HTML}{FFFFCB}%
\definecolor{listingFrame}{HTML}{BBBB98}%
\definecolor{listingLineno}{rgb}{0.5,0.5,1.0}%
\definecolor{LightGrey}{rgb}{0.975,0.975,0.975}
\lstdefinelanguage{btm}{
	commentstyle=\color{Gray},
	morecomment=[l]{//},
	morecomment=[s]{/*}{*/},
	classoffset=0,
        escapechar=\$,
	morekeywords={transaction,input,output,key,network,package},
	keywordstyle=\color{Blue}\bfseries,
	classoffset=1,
	morekeywords={sig,versig,fun,unit,int,string,bool,address,uint},
	keywordstyle=\color{TealBlue},
	classoffset=2,
	morekeywords={BTC,true},
	keywordstyle=\color{Plum}\bfseries,
}

\lstdefinelanguage{solidity}{
	commentstyle=\color{Gray},
	morecomment=[l]{//},
	morecomment=[s]{/*}{*/},
	classoffset=0,
        escapechar=\$,
	morekeywords={struct,mapping,function,this,public,private,static,final,class,extends,switch,case,break,finally,try,catch,return,if,else,new},
	keywordstyle=\color{Blue}\bfseries,
	classoffset=1,
	morekeywords={unit,int,string,bool,address,uint},
	keywordstyle=\color{TealBlue},
	classoffset=2,
	morekeywords={ether,wei,finney,contract,send,throw,msg,sender,value},
	keywordstyle=\color{Plum}\bfseries,
}

\lstdefinelanguage{java}{
	escapechar=\$,
        commentstyle=\color{Gray},
	morecomment=[l]{//},
	morecomment=[s]{/*}{*/},
	morestring=[b]",
        classoffset=0,
	morekeywords={public,private,static,final,class,extends,switch,case,break,finally,try,catch,void,int,boolean,throws,throw,return,if,else,new},
	keywordstyle=\color{keyword}\bfseries
}

%% file: macro.tex
\newcommand{\ifempty}[3]{%
  \ifthenelse{\isempty{#1}}{#2}{#3}%
}

\newcommand{\ifdots}[3]{%
  \ifthenelse{\equal{#1}{...}}{#2}{#3}%
}

\newcommand{\hidden}[1]{}

\newcommand{\mypar}[1]{\vspace{-0.2cm}\paragraph*{#1}}

\newenvironment{nscenter}
 {\parskip=0pt\par\nopagebreak\centering}
 {\parskip=2pt\par\noindent} % \ignorespacesafterend

\newcommand{\Real}[1]{\mathrm{Real}}

\def\etc{etc.\@\xspace}

\newcommand{\eg}{e.g.\@\xspace}
\newcommand{\ie}{i.e.\@\xspace}

\newcommand{\subst}[2]{\{{#1}/{#2}\}}

  {}

\newcommand{\sig}[3][]{\mathit{sig}^{#1}_{#2}\ifempty{#3}{}{({#3})}}

\newcommand{\BTC}{\textup{%
  \leavevmode
  \vtop{\offinterlineskip %\bfseries
    \setbox0=\hbox{B}%
    \setbox2=\hbox to\wd0{\hfil\hskip-.03em
    \vrule height .3ex width .15ex\hskip .08em
    \vrule height .3ex width .15ex\hfil}
    \vbox{\copy2\box0}\box2}}\xspace}

\def\pmvColor{\color{ForestGreen}}
\newcommand{\pmvFmt}[1]{{\pmvColor{\sf #1}}}

\newcommand{\PartT}{\pmvFmt{{Hon}}\xspace} % set of honest participants
\newcommand{\PartGC}{\pmvFmt{PartG}\xspace}  % set of participants in {G}C
\newcommand{\partC}[1]{\pmvFmt{part}({#1})\xspace} % participants in a contract

\newcommand{\pmv}[2][]{\pmvFmt{#2}_{\pmvColor{#1}}\xspace}
\newcommand{\pmvA}[1][]{\pmv[{#1}]{A}}
\newcommand{\pmvB}[1][]{\pmv[{#1}]{B}}
\newcommand{\pmvC}[1][]{\pmv[{#1}]{C}}

\newcommand{\pmvM}[1][]{\pmv[{#1}]{M}} % mediator
\newcommand{\pmvP}[1][]{\pmv[{#1}]{P}}
\newcommand{\Adv}{\pmv{Adv}} % adversary

\newcommand{\PmvP}[1][]{\pmvFmt{\mathcal{P}}_{#1}} % arbitrary set of participants

\def\txColor{\color{MidnightBlue}}
\newcommand{\txFmt}[1]{{\txColor{\sf #1}}}

\newcommand{\tx}[2][]{\txFmt{#2}_{\txColor{#1}}}
\newcommand{\txT}[1][]{\tx[#1]{T}}

\newcommand{\txIn}[2][]{{\small\textsf{in}}\ifempty{#1}{\ifempty{#2}{}{: {#2}}}{({#1})\ifempty{#2}{}{: {#2}}}}
\newcommand{\txWit}[2][]{{\small\textsf{wit}}\ifempty{#1}{\ifempty{#2}{}{: {#2}}}{({#1})\ifempty{#2}{}{: {#2}}}}
\newcommand{\txOut}[2][]{{\small\textsf{out}}\ifempty{#1}{\ifempty{#2}{}{: {#2}}}{({#1})\ifempty{#2}{}{: {#2}}}}
\newcommand{\txAfterAbs}[2][]{{\small\textsf{absLock}}\ifempty{#1}{\ifempty{#2}{}{: {#2}}}{({#1})\ifempty{#2}{}{: {#2}}}}

\DeclareMathAlphabet{\mathbfsf}{\encodingdefault}{\sfdefault}{bx}{n}

\newcommand{\fv}[1]{\ifempty{#1}{\operatorname{fv}}{\operatorname{fv}(#1)}}

\newcommand{\mmid}{\,\|\,}

\newcommand{\irule}[2]{\dfrac{#1}{#2}}

\newcommand{\bnfdef}{::=}
\newcommand{\bnfmid}{\;|\;}

\newcommand{\smallnrule}[1]{{\tiny \textsc{#1}}}
\newcommand{\sem}[2][]{\mbox{\ensuremath{\llbracket{#2}\rrbracket_{#1}}}}

\newcommand{\Nat}{\mathbb{N}}

\newcommand{\bind}[2]{\nicefrac{#2}{#1}}
\newcommand{\setenum}[1]{\{#1\}}

\newcommand{\qedef}{\ensuremath{\diamond}}

\crefname{appendix}{appendix}{appendices}
\Crefname{appendix}{Appendix}{Appendices}
\crefname{notation}{notation}{notations}
\Crefname{notation}{Notation}{Notations}

\definecolor{LightGrey}{rgb}{0.95,0.95,0.95}
\definecolor{keyword}{HTML}{7F0055}

\newlength\replength
\newcommand\repfrac{.1}

\newcommand\rulewidth{.6pt}
\newcommand\tdashfill[1][\repfrac]{\cleaders\hbox to \replength{%
  \smash{\rule[\arraystretch\ht\strutbox]{\repfrac\replength}{\rulewidth}}}\hfill}

\newcommand\tdotfill[1][\repfrac]{\cleaders\hbox to \replength{%
  \smash{\raisebox{\arraystretch\dimexpr\ht\strutbox-.1ex\relax}{.}}}\hfill}

\newcommand{\var}[2][]{#2_{#1}} % variables
\newcommand{\varX}[1][]{\var[#1]{x}} 
\renewcommand{\varXi}[1][]{\var[#1]{x'}} 
\newcommand{\varY}[1][]{\var[#1]{y}} 
 
\newcommand{\varZ}[1][]{\var[#1]{z}}
 
\newcommand{\varSig}[1][]{\var[#1]{\varsigma}} %change me

\newcommand{\varph}[2][]{#2_{#1}} % deposit place holder
\newcommand{\varphX}[1][]{\varph[#1]{d}} 
\newcommand{\varphXi}[1][]{\varph[#1]{d'}} 
\newcommand{\varphY}[1][]{\varph[#1]{e}} 
\newcommand{\varphYi}[1][]{\varph[#1]{e'}} 

\newcommand{\dmv}[1][]{\chi_{#1}}
\newcommand{\dmvi}[1][]{\chi'_{#1}}

\newcommand{\secr}[2][]{\mathord{{#2}_{#1}}}
\newcommand{\presecret}[2]{{#1}\!:\!\textup{\texttt{secret}}\,{\secr{#2}}}
\newcommand{\secrA}[1][]{\secr[#1]{a}} 
\newcommand{\secrAi}[1][]{\secr[#1]{a'}} 
\newcommand{\secrB}[1][]{\secr[#1]{b}} 

\newcommand{\secretC}[1]{\textup{\texttt{secrets}}({#1})\xspace}

\newcommand{\const}[2][]{#2_{#1}} % constants
\newcommand{\constK}[1][]{\const[#1]{k}}

\newcommand{\constT}[1][]{\const[#1]{t}}

\newcommand{\val}[2][]{#2_{#1}} % values
\newcommand{\valV}[1][]{\val[#1]{v}}
\newcommand{\valVi}[1][]{\val[#1]{v'}}

\newcommand{\expe}{\ensuremath{e}}

\newcommand{\versig}[2]{{\sf versig}_{#1}({#2})}

\newcommand{\true}{\mathit{true}}

\def\contrColor{\color{RubineRed}}

\newcommand{\contrFmt}[1]{{\contrColor{#1}}}
\newcommand{\contrC}[1][]{\mathord{\contrFmt{C}_{\contrColor{#1}}}}
\newcommand{\contrCi}[1][]{\mathord{\contrC[#1]\contrColor{'}}}
\newcommand{\contrCii}[1][]{\mathord{\contrC[#1]\contrColor{''}}}

\newcommand{\contrD}[1][]{\mathord{\contrFmt{D}_{\contrColor{#1}}}}
\newcommand{\contrDi}[1][]{\mathord{\contrD[#1]\contrColor{'}}}

\newcommand{\expE}[1][]{\ensuremath{E}_{#1}}

\newcommand{\contrG}[1][]{\mathord{\contrFmt{G}_{\contrColor{#1}}}}
\newcommand{\contrGi}[1][]{\mathord{\contrG[#1]\contrColor{'}}}

\newcommand{\predP}[1][]{\mathord{p_{#1}}}

\newcommand{\contrAdv}[3][]{\setenum{#2}^{#1}{#3}}  % symbolic contract advertisement
\newcommand{\contrAdvC}[2]{\mathcal{C}} % computational contract advertisement

\newcommand{\persdep}[3]{\mbox{\ensuremath{{#1}\textup{:}\,{#2}\,\textup{\texttt{@}}\,{#3}}}}
\newcommand{\assign}[3]{{#1}:{#2}\leftarrow{#3}}

\newcommand{\putname}{\textup{\texttt{put}}}
\newcommand{\revealname}{\textup{\texttt{reveal}}}
\newcommand{\wherename}{\textup{\texttt{if}}}
\newcommand{\andputname}{\&}

\newcommand{\putCtrue}[2]{%
  \ifempty{#2}
  {\putCnoreveal{#1}}
  {\ifempty{#1}
    {\revealname \, {#2}}
    {\putname \, {#1} \, \revealname \, {#2}}}
}

\newcommand{\putC}[3][]{
  \ifempty{#1}
  {\putCtrue{#2}{#3}}
  {\ifempty{#2}
    {\revealname \, {#3} \, \wherename \, {#1}}
    {\putname \, {#2} \, \andputname \, \revealname \, {#3} \, \wherename \, {#1}}
  }
}

\newcommand{\cnil}{0}

\newcommand{\splitname}{\textup{\texttt{split}}}
\newcommand{\splitC}[1]{\splitname \ifempty{#1}{}{\; #1}}
\newcommand{\splitB}[2]{{#1} \rightarrow {#2}}

\newcommand{\withdrawname}{\textup{\texttt{withdraw}}}
\newcommand{\withdrawC}[1]{\withdrawname\ifempty{#1}{}{\; {\pmv{#1}}}}

\newcommand{\authC}[2]{{\pmv{#1}}\,\textup{:}\,{#2}}

\newcommand{\aftername}{\texttt{after}}
\newcommand{\afterC}[2]{\textup{\aftername}\,{#1}\,\textup{:}\,{#2}}

\newcommand{\cVar}[2][]{\texttt{\textup{#2}}_{#1}}
\newcommand{\cVarX}[1][]{\cVar[#1]{X}}
\newcommand{\cVarY}[1][]{\cVar[#1]{Y}}

\newcommand{\procParamA}[1][]{\alpha_{#1}}
\newcommand{\procParamB}[1][]{\beta_{#1}}

\newcommand{\sexp}[1][]{\mathcal{E}_{#1}}
\newcommand{\sexpi}[1][]{\mathcal{E}'_{#1}}

\newcommand{\decl}[2]{{#1}({#2})}

\newcommand{\call}[2]{{#1}\langle{#2}\rangle}
\newcommand{\callX}[1]{\call{\cVarX}{#1}}
\newcommand{\callY}[1]{\call{\cVarY}{#1}}

\newcommand{\adv}[1]{{\pmvColor *}:\textup{\texttt{rngt}}\ifempty{#1}{}{\;}{#1}}
\newcommand{\extadv}[2]{{#1}:\textup{\texttt{rngt}}\ifempty{#1}{}{\;}{#2}}
\newcommand{\ncadv}[1]{\textup{\texttt{call}}\ifempty{#1}{}{\;}{#1}}

\newcommand{\confG}[1][]{\Gamma_{#1}}
\newcommand{\confGi}[1][]{\Gamma'_{#1}}
\newcommand{\confGii}[1][]{\Gamma''_{#1}}
\newcommand{\confGiii}[1][]{\Gamma'''_{#1}}
\newcommand{\confD}[1][]{\Delta_{#1}}

\newcommand{\confContr}[3][]{\langle {#2}, {#3} \rangle_{#1}}
\newcommand{\confDep}[3][]{\langle {#2}, {#3} \rangle_{#1}}
\newcommand{\confAuth}[2]{{#1}[{#2}]}
\newcommand{\confRev}[3]{{#1} : {#2}\#{#3}}
\newcommand{\confSec}[3]{\setenum{\confRev{#1}{#2}{#3}}}

\newcommand{\confTsep}{\mid}
\newcommand{\confT}[2]{{#1} \confTsep {#2}}

\newcommand{\authLab}[2]{{#1}:{\ifdots{#2}{\cdots}{#2}}} % authorization label

\newcommand{\authCommit}[3][]{\textbf{\#} \rhd \contrAdv[#1]{#2}{#3}}
\newcommand{\authAdv}[4][]{{#2} \rhd \contrAdv[#1]{#3}{#4}}
\newcommand{\authBranch}[2]{{#1} \rhd {#2}}
\newcommand{\authJoin}[4]{{#1},{#2} \rhd \confDep{#3}{#4}}

\def\compileColor{\color{Plum}}
\newcommand{\compileFmt}[1]{{\compileColor{\mathbf{#1}}}}

\newcommand{\valMap}{\compileFmt{val}}
\newcommand{\txMap}{\compileFmt{txout}}

\newcommand{\compilename}{\BTC}
\newcommand{\vectxT}[1][]{\mathcal{T}_{#1}}
\newcommand{\compile}[1]{\compileFmt{\compilename_{adv}}\ifempty{#1}{}{({#1})}}

\newcommand{\compileD}[1]{\compileFmt{\compilename_{D}}\ifempty{#1}{}{({#1})}}

\newcommand{\compileDscript}[1]{\compileFmt{\compilename_{out}}\ifempty{#1}{}{({#1})}}
\newcommand{\compileK}[1]{\compileFmt{K}({#1})}

\newcommand{\nonce}[1]{r_{#1}}   % nonce
\newcommand{\rndR}{r}            % randomness source

\newcommand{\runnameS}{\mathit{R}}
\newcommand{\runnameC}{\mathit{R}}

\newcommand{\runS}[1][]{\runnameS^{s}_{#1}}

\newcommand{\runC}[1][]{\runnameC^{c}_{#1}}

\newcommand{\stratS}[1]{\Sigma_{#1}^{\it s}}
\newcommand{\stratC}[1]{\Sigma_{#1}^{\it c}}

\newcommand{\stratSSet}{\mathbf{\Sigma}^{\it s}} % set of symbolic strategies
\newcommand{\stratCSet}{\mathbf{\Sigma}^{\it c}} % set of computational strategies

\newcommand{\stratMap}{\aleph}

\newcommand{\coherRel}[3]{{#1} \sim_{#3} {#2}}

%% file: abstract.tex
\begin{abstract}
  BitML is a process calculus to express smart contracts 
  that can be run on Bitcoin.
  One of its current limitations is that, once a contract has been
  stipulated, the participants cannot renegotiate its terms:
  this prevents expressing common financial contracts, where 
  funds have to be added by participants at run-time.
  In this paper, 
  we extend BitML with a new primitive for contract renegotiation.
  At the same time, the new primitive can be used to write recursive contracts,
  which was not possible in the original BitML.
  We show that, despite the increased expressiveness,
  it is still possible to execute BitML on standard Bitcoin,
  preserving the security guarantees of BitML.
\end{abstract}

%% file: intro.tex
\section{Introduction}
\label{sec:intro}

Smart contracts --- 
computer protocols that regulate the exchange of assets
in trustless environments ---
have become popular with the growth of interest in blockchain technologies.
Mainstream blockchain platforms like Ethereum, Libra, and Cardano, 
feature expressive high-level languages for programming smart contracts.
This flexibility has a drawback in that
it may open the door to attacks that steal or tamper with the
assets controlled by vulnerable contracts~\cite{ABC17post,Luu16ccs}.

An alternative approach, pursued first by Bitcoin 
and more recently also by Algorand,
is to sacrifice the expressiveness of smart contracts 
to reduce the attack surface.
For instance, Bitcoin has a minimal
language for transaction redeem scripts,
containing only a limited set of logic, arithmetic, and cryptographic operations.
Despite the limited expressiveness of these scripts, 
it is possible to encode a variety of smart contracts
(like gambling games, escrow services, crowdfunding systems, \etc)
by suitably chaining transactions 
\cite{Andrychowicz14bw,Andrychowicz14sp,Andrychowicz16cacm,bitcoinsok,Banasik16esorics,BZ17bw,Bentov14crypto,Kumaresan14ccs,KumaresanB16ccs,Kumaresan15ccs,KumaresanVV16ccs,Miller16zerocollateral}.
The common trait of these works is that they render contracts
as cryptographic protocols, where participants
can exchange/sign messages, read the blockchain, and append transactions.
Verifying the correctness of these protocols is hard, 
since it requires to reason in a computational model, 
where participants can manipulate arbitrary bitstrings,
only being constrained to use PPTIME algorithms.

Departing from this approach, BitML~\cite{BZ18bitml} 
allows to write Bitcoin contracts in a high-level, process-algebraic language.
BitML features a compiler that translates contracts 
into sets of standard Bitcoin transactions,
and a sound and complete verification technique
of relevant trace properties~\cite{BZ19post}.
The computational soundness of the compiler guarantees that 
the execution of the compiled contract is coherent with 
the semantics of the source BitML specification,
even in the presence of adversaries.
Although BitML can express many of the Bitcoin contracts 
presented in the literature, 
it is not ``Bitcoin-complete'', \ie there exist contracts 
that can be executed on Bitcoin, but are not expressible in BitML~\cite{bitmlracket}.

\newcommand{\ZCB}{\ensuremath{\contrFmt{\it ZCB}}\xspace}
\newcommand{\ZCBi}{\ensuremath{\contrFmt{\it ZCB2}}\xspace}
\newcommand{\ZCBii}{\ensuremath{\contrFmt{\it ZCB3}}\xspace}
\newcommand{\Claim}{\ensuremath{\contrFmt{\it Claim}}\xspace}

\newcommand{\PayFull}{\ensuremath{\contrFmt{\it PayFull}}\xspace}
\newcommand{\Pay}{\ensuremath{\contrFmt{\it Pay}}\xspace}
\newcommand{\PayG}{\ensuremath{\contrFmt{\it PayG}}\xspace}
\newcommand{\Escrow}{\ensuremath{\contrFmt{\it Escrow}}\xspace}
\newcommand{\Resolve}{\ensuremath{\contrFmt{\it Resolve}}\xspace}
\newcommand{\EscrowPut}{\ensuremath{\contrFmt{\it EscrowPut}}\xspace}

For instance, consider a zero-coupon bond~\cite{PeytonJones00icfp},
where an investor $\pmvA$ pays $1 \BTC$ upfront to a bank $\pmvB$,
and receives back $2 \BTC$ after a maturity date (say, year 2030).
We can express this contract in BitML as follows.
First, as a precondition to the stipulation of the contract, 
we require both $\pmvA$ and $\pmvB$ to provide a deposit:
$\pmvA$'s deposit is $1 \BTC$, while $\pmvB$'s deposit is $2 \BTC$.
In BitML, we write this precondition as:
\[
\persdep{\pmvA}{1 \BTC}{\varX[1]} 
\mid
\persdep{\pmvB}{2 \BTC}{\varX[2]}
\]
where $\varX[1]$ and $\varX[2]$ are the identifiers of transactions containing
the required amount of bitcoins ($\BTC$).
Under this precondition, 
we can specify the zero-coupon bond contract $\ZCB$ as follows:
\begin{align*}
  \ZCB 
  & \; = \;
    \splitname \; \big( 
    \splitB{1\BTC}{\withdrawC{\pmvB}}
    \; \mid \;
    \splitB{2\BTC}{\afterC{\text{2030}}{\withdrawC{\pmvA}}}
    \big)
\end{align*}

Upon stipulation, all the deposits required in the preconditions
pass under the control of $\ZCB$, 
and can no longer be spent by $\pmvA$ and $\pmvB$.
The contract splits these funds in two parts:
$1\BTC$, that can be withdrawn by $\pmvB$ at any moment,
and $2\BTC$, that can be withdrawn by $\pmvA$ after the maturity date.

Although $\ZCB$ correctly implements the functionality of zero-coupon bounds, 
it is quite impractical:
for the whole period from the stipulation to the maturity date,
$2\BTC$ are frozen within the contract, and cannot be used by the bank
in any way.
Although this is a desirable feature for the investor,
since it guarantees that he will receive $2\BTC$ even if the bank fails,
it is quite undesirable for the bank.
In the real world, the bank would be free to use its own funds,
together with those of investors, 
to make further financial transactions through which to repay the investments.
The risk that the bank fails is mitigated
by external mechanisms, like insurances or government intervention.

In this paper we propose an extension of BitML that overcomes this issue.
The idea is to allow the contract participants to \emph{renegotiate} it
after stipulation, in a controlled way.
Renegotiation makes it possible to inject in the contract new funds,
that were not specified in the original precondition.
We can use this feature to solve the issue with the $\ZCB$ contract.
The new precondition is
\(
\persdep{\pmvA}{1 \BTC}{\varX[1]} 
\),
\ie we only require $\pmvA$'s deposit. 
The revised contract is:
\begin{align*}
  \ZCBi
  & \; = \; \splitname \; \big( 
    \splitB{1\BTC}{\withdrawC{\pmvB}}
    \; \mid \;
    \splitB{0\BTC}{\adv{\callX{}}}
    \big)
  \\
  \callX{} 
  & \;= \; \contrAdv{\persdep{\pmvB}{2 \BTC}{\varphX}}{\;\afterC{\text{2030}}{\withdrawC{\pmvA}}}
\end{align*}

As before, the bank can withdraw $1\BTC$ at any moment after stipulation.
In the second part of the $\splitname$, 
the participants renegotiate the contract:
if they both agree, $0\BTC$ pass under the control of the contract $\callX{}$.
The precondition of $\callX{}$ requires the bank 
to provide $2\BTC$ in a fresh deposit;
upon renegotiation, $\pmvA$ can withdraw $2\BTC$ after the maturity date.
The crucial difference with $\ZCB$ is that the deposit variable $\varphX$
is instantiated at \emph{renegotiation} time, unlike $\varX$,
which must be fixed at \emph{stipulation} time.

The revised contract $\ZCBi$ solves the problem of $\ZCB$, 
in that it no longer freezes $2\BTC$ for the whole duration of the bond:
the bank could choose to renegotiate the contract,
paying $2\BTC$, just before the maturity date.
This flexibility comes at a cost, since $\pmvA$ loses
the guarantee to eventually receive $2\BTC$.
To address this issue we need to add, as in the real world,
an external mechanism.
More specifically, we assume an insurance company $\pmv{I}$ that,
for an annual premium of $p \BTC$ paid by the bank, 
covers a face amount of $f \BTC$ (with $2 > f > 10 p$):
\[
\persdep{\pmvA}{1 \BTC}{\varX[1]} 
\mid
\persdep{\pmvB}{p \BTC}{\varX[2]}
\mid
\persdep{\pmv{I}}{f \BTC}{\varX[3]}
\]
We revise the bond contract as follows:
\begin{align*}
  \ZCBii
  & \; = \; \splitname \; \big( 
    \splitB{1 \BTC}{\withdrawC{\pmvB}}
  \\
  & \hspace{44pt} \mid\hspace{1pt}
    \splitB{p \BTC}{\withdrawC{\pmv{I}}}
  \\
  & \hspace{44pt} \mid
    \splitB{f \BTC}{\adv{\callX{1}} + \afterC{2021}{\withdrawC{\pmvA}}}
    \big)
  \\
  \callX{n \in 1..9} 
  & \;= \; \contrAdv{\persdep{\pmvB}{p \BTC}{\varphX}}{}
  \\
  & \hspace{20pt} \splitname \; \big( \splitB{p\BTC}{\withdrawC{\pmv{I}}}
  \\
  & \hspace{44pt} \mid
  \splitB{f\BTC}{\adv{\callX{n+1}} + \afterC{(2021+n)}{\withdrawC{\pmvA}}}
  \big)
  \\
  \callX{10} 
  & \;= \; \contrAdv{\persdep{\pmvB}{2 \BTC}{\varphX}}{}
  \\
  & \hspace{20pt} \splitname \; \big( \splitB{f\BTC}{\withdrawC{\pmv{I}}}
  \\
  & \hspace{44pt} \mid
    \splitB{2\BTC}{\afterC{\text{2030}}{\withdrawC{\pmvA}}}
    \big)
\end{align*}

The contract starts by transferring $1\BTC$ to the bank,
and the first year of the premium to the insurer.
The remaining $f\BTC$ are transferred to the renegotiated contract
$\callX{1}$, or, if the renegotiation is not completed by 2021,
to the investor.

We remark that the pattern $\contrD + \afterC{t}{\contrDi}$,
where $\contrD$ requires some authorizations but $\contrDi$ does not,
is rather common in BitML,
as it ensures that the contract can proceed
even if the authorizations are not provided.
Indeed, in such case an honest participant is enough to
execute $\contrDi$ after time $t$.
By suitably exploiting this pattern, it is possible to guarantee
that a BitML contract enjoys liveness, 
by just assuming that at least one participant is honest.

The contracts $\callX{n}$, for $n \in 1..9$, 
allow the insurer to receive the annual premium until 2030:
if the bank does not renegotiate the contract for the following year
(paying the corresponding premium), then
the investor can redeem the face amount of $f \BTC$.
Finally, the contract $\callX{10}$ can be triggered if the bank deposits
the $2\BTC$: when this happens, the face amount is given back to the insurer,
and the investor can redeem $2\BTC$ after the maturity date.
 
Compared to $\ZCBi$, the contract $\ZCBii$ offers more protection 
to the investor.
To see why, we must evaluate $\pmvA$'s payoff for all the possible 
behaviours of the other participants.
If $\pmvB$ and $\pmv{I}$ are both honest,
then $\pmvA$ will redeem $2 \BTC$, as in the ideal contract $\ZCB$.
Instead,
if either $\pmvB$ or $\pmv{I}$ do not accept to renegotiate some $\callX{n}$, 
% for some $n \in 1..10$,
then $\pmvA$ can redeem $f \BTC$ as a partial compensation
(unlike in $\ZCBi$, where $\pmvA$ just loses $1\BTC$).
% (after $2021+n$) 
In the real world, $\pmvA$ could use this compensation
to cover the legal fee to sue the bank in court;
also, $\pmv{I}$ could \eg increase the premium for future interactions 
with $\pmvB$.
By further refining the contract, 
we could model these real-world mechanisms as oracles, 
which sanction dishonest participants
according to the evidence collected in the blockchain
and in messages broadcast by participants.
For instance, if $\pmvB$ and $\pmv{I}$ accept the renegotiation $\callX{n}$
but $\pmvA$ does not, 
then the oracle would be able to detect $\pmvA$'s dishonesty
by inspecting the authorizations broadcast in year $2021+n$.
The sanction could consist \eg in blacklisting $\pmvA$, 
so to prevent her from buying other bonds from $\pmvB$.

\mypar{Contributions.}
We summarise our main contributions as follows:
\begin{itemize}

\item We extend BitML with the renegotiation primitive $\adv{\callX{}}$,
  suitably adapting the language syntax and semantics.
  The new primitive increases the expressiveness of BitML:
  besides allowing participants to provide new deposits and secrets
  at run-time, it also allows for \emph{unbounded} recursion.

\item We extend the BitML compiler to the new primitive,
  making it possible to execute renegotiations on Bitcoin.
  We accordingly extend the computational soundness result in~\cite{BZ18bitml},
  guaranteeing that the BitML semantics is coherent with the
  actual Bitcoin executions, also in the presence of adversaries.

\item We exploit renegotiation 
  to design a new gambling game where two players repeatedly flip coins,
  and whoever wins twice in a row takes the pot
  (a form of unbounded recursion).
  We prove the game to be fair.

\item We introduce alternative renegotiation primitives,
  which allow participants to choose some parameters 
  (\eg the amounts to be deposited) at renegotiation time,
  and to change the set of participants involved in the renegotiated contract.
  We show that both primitives can be executed on Bitcoin \emph{as is}.
  We also introduce a primitive that, at the price of minor Bitcoin extensions, 
  supports \emph{non-consensual} renegotiations, 
  which are automatically triggered by the
  contract without requiring the participants’ agreement.

\end{itemize}

\noindent
Because of space constraints, we relegate part of the technicalities to~\Cref{app}.

%% file: bitml.tex
\section{BitML with renegotiation and recursion}
\label{sec:bitml}

We start by formalising contract preconditions.
We use $\pmvA, \pmvB, \ldots$ to range over participants.
We assume a set of \emph{deposit names} $\varX, \varY, \ldots$,
a set of deposit variables $\varphX,\varphY,\hdots$,
and a set of \emph{secret names} $\secrA, \secrB, \ldots$.
We use $\dmv,\dmvi,\hdots$ to range over deposit names and variables,
and $\valV, \valVi$ to range over non-negative values.

\begin{definition}[\textbf{Contract precondition}]
  \label{def:bitml:precondition}
  Contract preconditions have the following syntax
  (the deposits $\dmv$ in a contract precondition $\contrG$ must be distinct):
  \begin{align*}
    \contrG \bnfdef \;\;
    % & \tempdep{\pmvA}{\valV}{\dmv}
    % && \text{volatile deposit of $\valV \BTC$, expected from $\pmvA$}
    % \\
    % \bnfmid
      & \persdep{\pmvA}{\valV}{\dmv}
      && \text{deposit of $\valV \BTC$ put by $\pmvA$}
    \\
    \bnfmid\;
      & \presecret{\pmvA}{\secrA} 
      && \text{secret committed by $\pmvA$}
    \\
    \bnfmid\;
      & \contrG \mid \contrG
      && \text{composition}
         \tag*{$\qedef$}
  \end{align*}
\end{definition}

The precondition $\persdep{\pmvA}{\valV}{\dmv}$ 
requires $\pmvA$ to own $\valV \BTC$ in a deposit $\dmv$,
and to spend it for stipulating the contract.
The precondition $\presecret{\pmvA}{\secrA}$ requires $\pmvA$ to generate a
secret $\secrA$, and commit to it before the contract starts. 
After stipulation, $\pmvA$ can choose whether to disclose 
the secret $\secrA$, or not.

To define contracts,
we assume a set of recursion variables, ranged over by
$\cVarX,\cVarY,\hdots$,
and a language of \emph{static expressions} $\sexp,\sexpi,\hdots$, 
formed by integer constants $k$, 
integer variables $\procParamA,\procParamB,\hdots$,
and the usual arithmetic operators.
We omit to define the syntax and semantics of static
expressions, since they are standard.
We assume that a closed static expression evaluates to a 32-bit value.
We use the bold notation for sequences, \eg 
$\vec{\varX}$ denotes a finite sequence of deposit names.

\begin{definition}[\textbf{Contract}]
  \label{def:bitml:contract}
  Contracts are terms with the syntax in~\Cref{fig:bitml:contract}, where:
  \begin{inlinelist}
  \item each recursion variable $\cVarX$ has a unique defining equation 
    $\decl{\cVarX}{\vec{\procParamA}} = \contrAdv{\contrG}{\contrC}$;
    \item renegotiations $\adv{\callX{\vec{\sexp}}}$ 
    have the correct number of arguments;
  \item the names $\vec{\secrA}$ in
    $\putC[p]{}{\vec{\secrA}}$ are distinct, 
    and they include those occurring in $p$;
  \item in a prefix $\,\splitC{\splitB{\vec{\valV}}{\vec{\contrC}}}$,
    the sequences $\vec{\valV}$ and $\vec{\contrC}$ have the same length.
  \end{inlinelist}
  We denote with $\cnil$ the empty sum. 
  We assume that the order of decorations 
  is immaterial,
  \eg, $\afterC{\sexp}{\authC{A}{\authC{B}{\contrD}}}$ is equivalent to
  $\authC{B}{\authC{A}{\afterC{\sexp}{\contrD}}}$.
  \hfill{$\qedef$}
\end{definition}

A contract $\contrC$ is a choice among guarded contracts $\contrD[i]$. 
A guarded contract $\putC[\predP]{}{\vec{\secrA}} .\, \contrCi$ 
continues as $\contrCi$
once all the secrets $\vec{\secrA}$ have been revealed
and satisfy the predicate $\predP$.
The guarded contract 
\mbox{$\splitC{(\splitB{\valV[1]}{\contrC[1]} \mid \cdots \mid \splitB{\valV[n]}{\contrC[n]})}$}
divides the contract into $n$ contracts $\contrC[i]$, 
each one with balance $\valV[i]$. 
The sum of the $\valV[i]$ must coincide with the current balance.
The action $\withdrawC{A}$ transfers the whole balance to $\pmvA$.
When enabled, the above actions can be fired by anyone at anytime. 
To restrict \emph{who} can execute a branch and \emph{when}, 
one can use the decoration $\authC{\pmvA}{\contrD}$, 
requiring to wait for $\pmvA$'s authorization,
and the decoration $\afterC{\sexp}{\contrD}$,
requiring to wait until the time specified by the static
expression $\sexp$.
The action $\adv{\callX{\vec{\sexp}}}$ allows the participants 
involved in the contract to renegotiate it.
Intuitively, if
$\decl{\cVarX}{\vec{\procParamA}} = \contrAdv{\contrG}{\contrC}$,
then the contract continues as 
$\contrC\setenum{\bind{\vec{\procParamA}}{\vec{\sexp}}}$  
if all the participants give their authorization,
and satisfy the precondition $\contrG$.

\begin{figure*}[t]
  \small
  \begin{minipage}{0.62\textwidth}
    \begin{align*}
      % \gContrP & \bnfdef \contrAdv{\contrG}{\contrC} && \text{advertised contract}
      % \\
      % \bnfmid & \callX{\vec{\sexp}}
      % && \text{call}
      % \\[4pt]
      \contrC 
                 & \bnfdef 
                   \textstyle \sum_{i \in I} \contrD[i]
                                                       && \text{contract}
      \\[4pt]
      \contrD
                 & \bnfdef
                                                       && \text{guarded contract}
      \\
                 & \putC[\predP]{}{\vec{\secrA}} . \, \contrC
                                                       && \text{reveal secrets (if $\predP$ is true)} 
      \\[2pt]
      \bnfmid
                 & \withdrawC{\pmvA}
                                                       && \text{transfer the balance to $\pmvA$}
      \\[2pt]
      \bnfmid
                 & \splitC{\splitB{\vec{\valV}}{\vec{\contrC}}}
                                                       && \text{split the balance}
      \\[2pt]
      \bnfmid
                 & \authC{\pmvA}{\contrD}
                                                       && \text{wait for $\pmvA$'s authorization}
      \\[2pt]
      \bnfmid
                 & \afterC{\sexp}{\contrD}
                                                       && \text{wait until time $\sexp$}
      \\[2pt]
      \bnfmid
                 & \adv{\callX{\vec{\sexp}}}
                                                       && \text{renegotiate the contract}
    \end{align*}
  \end{minipage}
  \begin{minipage}{0.35\textwidth}
    \begin{align*}
      \predP 
      \bnfdef \; &
      % && \text{predicate}
           \true
      && \text{truth}
      \\
      \bnfmid
                 & \predP \land \predP
        && \text{conjunction}
      \\
      \bnfmid
                 & \neg \predP
        && \text{negation}
      \\
      \bnfmid
                 & \expE = \expE 
        && \text{equality} 
      \\
      \bnfmid
                 & \expE < \expE 
        && \text{less than}
      \\
      \expE
      \bnfdef \; &
      % && \text{contract expression}
           \sexp
      && \text{static expression}
      \\
      \bnfmid
                 & \secrA
        && \text{secret}
      \\
      \bnfmid
                 & \expE + \expE
        && \text{addition} 
      \\
      \bnfmid
                 & \expE - \expE
        && \text{subtraction}
    \end{align*}
  \end{minipage}
  % \vspace{-5pt}
  \caption{Syntax of BitML contracts.}
  \label{fig:bitml:contract}
\end{figure*}

\begin{definition}[\textbf{Contract advertisement}]
  \label{def:bitml:contrAdv}
  A contract advertisement is a term $\contrAdv{\contrG}{\contrC}$ 
  such that:
  \begin{inlinelist}
  \item each secret name in $\contrC$ occurs in $\contrG$; 
  \item \label{item:bitml:contrAdv:persistent}
    $\contrG$ requires a deposit from each 
    $\pmvA$ in $\contrAdv{\contrG}{\contrC}$;
  \item each $\adv{\callX{\vec{\sexp}}}$ in $\contrC$ refers to
    a defining equation
    $\decl{\cVarX}{\vec{\procParamA}} = \contrAdv{\contrGi}{\contrCi}$
    where the participants in $\contrGi$ are the same as those in
    $\contrG$.
    \hfill{$\qedef$}
  \end{inlinelist}
\end{definition}

The second condition is used to guarantee that the contract is stipulated 
only if \emph{all} the involved participants give their authorizations.
The last condition is only used to simplify the technical development.
We outline in~\Cref{sec:discussion} how to relax it, 
by allowing renegotiations to exclude some participants,
or to include new ones, which were not among those who originally 
stipulated the contract.

We now extend the reduction semantics of BitML~\cite{BZ18bitml},
by focussing on the new renegotiation primitive.
Because of space limitations, here we just provide the underlying intuition,
relegating the full formalisation to~\Cref{app}.
We start by defining the configurations of the semantics.

\begin{definition}[\textbf{Configuration}]
  \label{def:bitml:conf}
  Configurations have the following syntax:
  \begin{align*}
    \confG \bnfdef \;\;
    & \cnil
    && \text{empty} 
    \\
    \bnfmid
    & \contrAdv[\varX]{\contrG}{\contrC}
    && \text{contract advertisement (name $\varX$ is optional)}
    \\
    \bnfmid
    & \confContr[x]{\contrC}{\valV}
    && \text{active contract containing $\valV \BTC$}
    \\
    \bnfmid
    & \confDep[x]{\pmvA}{\valV}
    && \text{deposit of $\valV \BTC$ redeemable by $\pmvA$}
    \\
    \bnfmid
    & \confAuth{\pmvA}{\chi}
    && \text{authorization of $\pmvA$ to perform action $\chi$}
    \\
    \bnfmid
    & \confSec{\pmvA}{\secrA}{N}
    && \text{committed secret of $\pmvA$ ($N \in \Nat \cup \setenum{\bot}$)}
    \\
    \bnfmid
    & \confRev{\pmvA}{\secrA}{N}
    && \text{revealed secret of $\pmvA$ ($N \in \Nat$)}
    \\
    \bnfmid
    & \assign{\pmvA}{\varphX}{\varX}
    && \text{$\pmvA$'s deposit variable $\varphX$ assigned to deposit name $\varX$}
    \\
    \bnfmid
    & \confG \mid \confGi
    && \text{parallel composition}    
  \end{align*}
  We denote with $\confT{\confG}{\constT}$
  a \emph{timed} configuration, 
  where $\constT \in \Nat$ is a global time.
  \hfill\qedef
\end{definition}

We illustrate configurations and their semantics through a series of examples. 

\mypar{Deposits.}
A deposit $\confDep[x]{\pmvA}{\valV}$ can be subject to several operations,
like \eg split into two smaller deposits, join with another deposit,
transfer to another participant, or destroy. 
In all cases $\pmvA$ must authorise the operation.
For instance, to authorize the join of two deposits, 
$\pmvA$ can perform the following step:
\[
\small
\confDep[\varX]{\pmvA}{\valV} \mid \confDep[\varY]{\pmvA}{\valVi} 
% \mid \confG
\; \xrightarrow{} \;
\confDep[\varX]{\pmvA}{\valV} \mid \confDep[\varY]{\pmvA}{\valVi} \mid
\confAuth{\pmvA}{\authJoin{\varX}{\varY}{\pmvA}{\valV + \valVi}}  
% \mid \confG
\]
where $\authJoin{\varX}{\varY}{\pmvA}{\valV + \valVi}$ is
the authorization of $\pmvA$ to spend $\varX$.
After $\pmvA$ also provides the dual authorization to spend $\varY$,
% rule \nrule{[Dep-Join]} 
anyone can actually join the deposits:
\[
\small
\confDep[x]{\pmvA}{\valV} \mid \confDep[y]{\pmvA}{\valVi} 
\mid
\confAuth{\pmvA}{\authJoin{\varX}{\varY}{\pmvA}{\valV + \valVi}} \mid 
\confAuth{\pmvA}{\authJoin{\varY}{\varX}{\pmvA}{\valV + \valVi}}
\; \xrightarrow{} \;
\confDep[z]{\pmvA}{\valV + \valVi} 
\]

\mypar{Advertisement.}
Any participant can broadcast a new contract advertisement
$\contrAdv{\contrG}{\contrC}$,
provided that all the deposits mentioned in $\contrG$ exist 
in the current configuration, 
and that the names of the secrets declared in $\contrG$ are fresh.

\mypar{Stipulation.}
To stipulate an advertised contract $\contrAdv{\contrG}{\contrC}$, 
all the participants mentioned in it
must fulfill the preconditions, and authorise the stipulation. 
For instance, let 
$\contrG = \persdep{\pmvA}{1}{x} \mid \persdep{\pmvB}{1}{y} \mid \presecret{\pmvA}{\secrA}$,
and let $\contrC$ be an arbitrary contract involving only $\pmvA$ and $\pmvB$.
The stipulation starts from a configuration
containing the advertisement and the participants' deposits:
\[
  \confG 
  \;\; = \;\; 
  \contrAdv{\contrG}{\contrC} \mid
  \confDep[x]{\pmvA}{1} \mid \confDep[y]{\pmvB}{1}
\]
At this point the participants must commit to their secrets
(in this case, only $\pmvA$ has a secret). 
This is rendered as a sequence of steps:
\[
\confG
\; \xrightarrow{}^* \;
\confG
\mid \confSec{\pmvA}{\secrA}{N}
\mid \confAuth{\pmvA}{\authCommit{\contrG}{\contrC}}
\mid \confAuth{\pmvB}{\authCommit{\contrG}{\contrC}}
\; = \; \confGi
\]
where $\confSec{\pmvA}{\secrA}{N}$ represents the fact that 
$\pmvA$ has committed to the secret $N$,
while 
$\confAuth{\pmvA}{\authCommit{\contrG}{\contrC}}$ 
and 
$\confAuth{\pmvB}{\authCommit{\contrG}{\contrC}}$ 
represent ending the commitment phase
(these steps might seem redundant, but they are useful
to obtain a step-by-step correspondence between BitML executions
and Bitcoin executions).

After that, $\pmvA$ and $\pmvB$ must perform an additional sequence of steps
to authorize the transfer of their deposits $\varX$, $\varY$ to the contract:
\[
\confGi
\; \xrightarrow{}^* \;
\confGi
\mid \confAuth{\pmvA}{\authAdv{\varX}{\contrG}{\contrC}}
\mid \confAuth{\pmvB}{\authAdv{\varY}{\contrG}{\contrC}}
\; = \; \confGii
\]
where $\confAuth{\pmvA}{\authAdv{x}{\contrG}{\contrC}}$ and 
$\confAuth{\pmvB}{\authAdv{y}{\contrG}{\contrC}}$ are
the authorizations to spend $\varX$ and $\varY$.

At this point all the needed authorizations have been given, 
so the advertisement can be turned into an active contract. 
This step consumes the deposits and all the authorizations,
and creates an active contract, with a fresh name $\varZ$:
\[
\confGii
\; \xrightarrow{} \;
\confContr[z]{\contrC}{2} \mid \confSec{\pmvA}{\secrA}{N} 
\]

\mypar{Renegotiation.}

We illustrate the steps to renegotiate
$\callX{\procParamA} = \contrAdv{\contrG}{\contrC}$, where
\(
\contrG = 
\persdep{\pmvA}{1}{\varphX} \mid
\persdep{\pmvB}{1}{\varphY} \mid
\presecret{\pmvA}{\secrA}
\),
and $\contrC$ is an arbitrary contract involving only $\pmvA$ and $\pmvB$, and 
possibly containing the integer variable $\procParamA$ in static expressions.
Here, $\contrG$ requires $\pmvA$ and $\pmvB$ to spend two $1\BTC$ deposits,
and $\pmvA$ to commit to a secret.
Unlike in the case of contract stipulation above,
deposits names are unknown before renegotiation, 
so we use the deposit variables $\varphX,\varphY$ to refer to them.

Consider a configuration
$\confContr[\varX]{\adv{\callX{\constK}} + \contrCii}{\valV} \mid
\confG$, where $\contrCii$ contains the branches alternative to
the renegotiation. % $\adv{\callX{\constK}}$.
A possible execution of the action $\adv{\callX{\constK}}$ starts as follows:
\[
\confContr[\varX]{\adv{\callX{\constK}} + \contrCii}{\valV} \mid \confG 
\; \xrightarrow{} \;
\confContr[\varX]{\adv{\callX{\constK}} + \contrCii}{\valV} 
\mid \contrAdv[\varX]{\contrGi}{\contrCi} \mid \confG
\; = \; \confGi
\]
where the advertisement $\contrAdv[\varX]{\contrGi}{\contrCi}$ is
obtained by transforming $\contrAdv{\contrG}{\contrC}$ as follows:
\begin{inlinelist}
\item variables $\varphX,\varphY$ are
  renamed into fresh ones $\varphXi,\varphYi$,
  and similarly the secret name $\secrA$ into $\secrAi$,
  %(secret names are also similarly renamed),
  % 
\item the static expressions in $\contrC$ are evaluated, assuming
  $\procParamA = \constK$, and replaced with their results.
\end{inlinelist}
The superscript $\varX$ in the advertisement is used to record that,
when the renegotiation is concluded, the contract $\varX$ must be reduced.

In the subsequent steps participants choose the actual deposit names, 
and $\pmvA$ commits to her secret.
If $\pmvA$ owns in $\confG$ a deposit $\confDep[\varY]{\pmvA}{1}$,
she can choose $\varphXi = \varY$ to satisfy the precondition $\contrG$.
Similarly, $\pmvB$ can choose $\varphYi = \varZ$ if he owns such a deposit
in $\confG$.
These choices are performed as follows:
\begin{align*}
  \confGi 
  \; \xrightarrow{}^* \;
  \confGi 
  & \mid 
    \assign{\pmvA}{\varphXi}{\varY} \mid 
    \confSec{\pmvA}{\secrAi}{N} \mid \confAuth{\pmvA}{\authCommit[\varX]{\contrGi}
    {\contrCi}}
  \\ 
  & \mid \assign{\pmvB}{\varphYi}{\varZ} \mid \confAuth{\pmvB}
    {\authCommit[\varX]{\contrGi}{\contrCi}}
    \hspace{60pt} = \; \confGii                   
\end{align*}

At this point, participants must authorise spending their deposits
and the balance of the contract at $\varX$.
This is done through a series of steps:
\begin{align*}
  \confGii 
  \; \xrightarrow{}^* \;
  \confGii 
  & \mid \confAuth{\pmvA}{\authAdv[\varX]
    {\varY}{\contrGi}{\contrCi}} \mid 
    \confAuth{\pmvA}{\authAdv[\varX]{\varX}{\contrGi}{\contrCi}}
  \\ 
  & \mid 
  \confAuth{\pmvB}
    {\authAdv[\varX]{\varZ}{\contrGi}{\contrCi}} \mid 
    \confAuth{\pmvB}{\authAdv[\varX]{\varX}{\contrGi}{\contrCi}}
    \; = \; \confGiii
\end{align*}

Finally, the new contract is stipulated. 
This closes the old contract, 
and transfers its balance to the newly generated one,
with a fresh name $\varXi$:
\[
\confGiii 
\; \xrightarrow{} \;
\confContr[\varXi]{\contrCi}{\valV + 2} 
\mid \confG
\]

Note that the branches in $\contrCii$ are discarded only in the last step
above, where we complete the renegotiation.
Before this step, it would have been possible to take one of the
branches in $\contrCii$, aborting the renegotiation.

\mypar{Withdraw.}
Executing $\withdrawC{\pmvA}$ transfers the whole contract balance to~$\pmvA$:
\[
\confContr[x]{\withdrawC{\pmvA} + \contrCi}{\valV}   
\;\xrightarrow{}\;
\confDep[y]{\pmvA}{\valV} 
\]
After the execution, the alternative branch $\contrCi$ is discarded, 
and a fresh deposit of $\valV\BTC$ for $\pmvA$ is created.
Note that the active contract $\varX$ is terminated.

\mypar{Split.}
The $\splitname$ primitive divides the contract balance in $n$ parts,
each one controlled by its own contract.
For instance, if $n=2$:
\[
\confContr[x]
{(\splitC{\splitB{\valV[1]}{\contrC[1]} \mid \splitB{\valV[2]}{\contrC[2]}})
  + \contrCi}
{\valV[1]+\valV[2]} 
\;\xrightarrow{}\;
\confContr[y]{\contrC[1]}{\valV[1]} \mid 
\confContr[z]{\contrC[2]}{\valV[2]}
\]
After this step, the new spawned contracts $\contrC[1]$ and $\contrC[2]$
are executed concurrently.

\mypar{Reveal.}
The prefix $\putC[\predP]{}{\vec{\secrA}}$ can be fired if all the  
committed secrets $\vec{\secrA}$ have been revealed, 
and satisfy the guard $\predP$.
For instance, 
if $\confG = \confRev{\pmvA}{\secrA}{N} \mid \confRev{\pmvB}{\secrB}{N}$:
\[
  \confContr[x]
  {(\putC[\secrA=\secrB]{}{\secrA\secrB}.\, \contrC) + \contrCi}
  {\valV}
  \mid 
  \confG
  \; \xrightarrow{} \;
  \confContr[y]{\contrC}{\valV}
  \mid 
  \confG
\]
The terms $\confRev{\pmvA}{\secrA}{N}$ and $\confRev{\pmvB}{\secrB}{N}$
represent the fact that the secrets $\secrA$ and $\secrB$ have been revealed.
Crucially, only the participant who performed the commitment 
can add the corresponding term to the configuration.

\mypar{Authorizations.}
A branch decorated by $\authC{\pmvA}{\cdots}$ can be taken only 
if the participant $\pmvA$ has provided her authorization.
For instance:
\[
\confContr[x]
{\authC{\pmvA}{\withdrawC{\pmvB} + \contrCi}}
{\valV}
\mid
\confAuth{\pmvA}{\authBranch{x}{\authC{\pmvA}{\withdrawC{\pmvB}}}}
\;\xrightarrow{}\;
\confDep[y]{\pmvB}{\valV}
\]
The leftmost configuration contains the term
$\confAuth{\pmvA}{\authBranch{x}{\authC{\pmvA}{\withdrawC{\pmvB}}}}$,
which represents $\pmvA$'s authorization to take the branch 
$\withdrawC{\pmvB}$. 
This enables the step to be taken.
When multiple authorizations are required, 
the branch can be taken only after all of them occur in the configuration.

\mypar{Time constraints.}
We represent time in configurations as $\confT{\confG}{\constT}$,
where $\confG$ is the untimed part of the configuration and $\constT$ 
is the current time.
We always allow the time to advance through the rule
\(
\confT{\confG}{\constT}
\xrightarrow{}
\confT{\confG}{\constT+\delta}
\),
for all $\delta>0$.
A branch decorated with $\aftername \, d$ 
can be taken only if time $d$ has passed.
For instance:
\[
\confT
{\confContr[x]
  {\afterC{d}{\withdrawC{\pmvB}}}
  {\valV}}
{\constT}
\; \xrightarrow{} \;
\confT{\confDep[y]{\pmvB}{\valV}}{\constT}
\tag*{if $t \geq d$}
\]

\noindent
For the branches not guarded by $\aftername$,
we lift transitions from untimed to timed configurations:
namely, for an untimed transition $\confG \xrightarrow{} \confGi$, 
we also have the timed transition 
$\confT{\confG}{\constT} \xrightarrow{} \confT{\confGi}{\constT}$.
This reflects the assumption that participants 
can always meet deadlines, if they want to.

%% file: compiler.tex
\section{Executing BitML on Bitcoin}
\label{sec:compiler}

To execute a BitML contract, participants first compile it to 
a set of Bitcoin transactions,
and then append these transactions to the blockchain,
each following their own strategy.
Participants' strategies can involve other actions
besides appending transactions, 
like \eg broadcasting signatures on given transactions
(which corresponds, in BitML, to add an authorization to the configuration),
revealing secrets, and waiting some time
(see Definition 16 in~\cite{BZ18bitml}).
The coherence between the BitML semantics and the execution on Bitcoin 
is guaranteed by a step-by-step correspondence between 
the transitions of the BitML semantics 
and the actions performed by participants on the Bitcoin network.

In this~\namecref{sec:compiler} we illustrate the compiler and the 
execution protocol through a couple of examples, 
focussing on the new renegotiation primitive.
The needed background on Bitcoin will be introduced along with
these examples.
We relegate the formal definition of the compilation rules to~\Cref{app:compiler}.

\paragraph{Zero-coupon bond.}
Recall the $\ZCB$ contract from~\Cref{sec:intro}:
\begin{align*}
  \ZCB 
  & \; = \;
    \splitname \; \big( 
    \splitB{1}{\withdrawC{\pmvB}}
    \; \mid \;
    \splitB{2}{\afterC{\text{2030}}{\withdrawC{\pmvA}}}
    \big)
\end{align*}
The precondition  
$\persdep{\pmvA}{1}{\varX[1]} \mid \persdep{\pmvB}{2}{\varX[2]}$
requires $\pmvA$ to deposit $1\BTC$ in the contract, 
and $\pmvB$ to deposit $2\BTC$.
In Bitcoin, this precondition corresponds to requiring
two unspent transactions redeemable by $\pmvA$ and $\pmvB$,
and containing the required amounts.
We represent these transactions as follows, 
using the notation in~\cite{bitcointxm}:
\begin{nscenter}
  \small
  \begin{tabular}[t]{|l|}
    \hline
    \multicolumn{1}{|c|}{$\txT[x_1]$} \\
    \hline
    \txIn{$\cdots$} \\
    \txWit{$\cdots$} \\
    \txOut{$(\lambda \varX. \versig{\compileK{\pmvA}}{\varX}, 1 \BTC)$} \\
    \hline
  \end{tabular}
  \qquad
  \begin{tabular}[t]{|l|}
    \hline
    \multicolumn{1}{|c|}{$\txT[x_2]$} \\
    \hline
    \txIn{$\cdots$} \\
    \txWit{$\cdots$} \\
    \txOut{$(\lambda \varX. \versig{\compileK{\pmvB}}{\varX}, 2 \BTC)$} \\
    \hline
  \end{tabular}
\end{nscenter}

The transaction $\txT[x_1]$ is a record with three fields
($\txT[x_2]$ is similar).
The $\txIn{}$ field points to one or more previous transactions
in the blockchain.
The field $\txOut{}$ is a pair, whose first element is a boolean predicate
(with parameter $\varX$),
and the second element, $1 \BTC$, is the amount that a subsequent
transaction satisfying the predicate can redeem from $\txT[x_1]$.
Here, the predicate $\versig{\compileK{\pmvA}}{\varX}$ is true when
$\varX$ is a signature of $\pmvA$ on the redeeming transaction
(\ie, one having $\txT[x_1]$ as $\txIn{}$).

\begin{figure}[t]
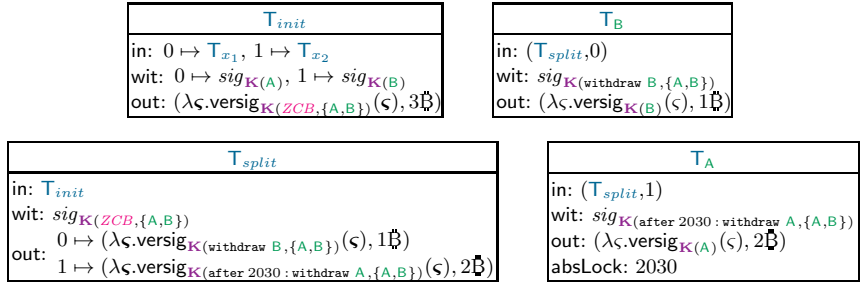

  \centering
  \resizebox{0.95\columnwidth}{!}{
    % \hspace{-10pt}
    \small
    \begin{tabular}{c}
      % 
      %%%%%%%%%%%%%%%%%%%%%%%%%%%%%%%%%%%%%%%%%%%%%%%%%%%%%%%%%%%% 
      % Tinit
      %%%%%%%%%%%%%%%%%%%%%%%%%%%%%%%%%%%%%%%%%%%%%%%%%%%%%%%%%%%% 
      \begin{tabular}{|l|}
        \hline
        \\[-9pt]
        \multicolumn{1}{|c|}{$\txT[init]$} \\
        \hline
        \\[-9pt]
        \txIn{$\begin{array}{l}
                 0 \mapsto \txT[\txColor{\varX[{1}]}],\,
                 1 \mapsto \txT[\txColor{\varX[{2}]}]
               \end{array}$} \\
        \txWit{$\begin{array}{l}
                  0 \mapsto \sig{\compileK{\pmvA}}{},\,
                  1 \mapsto \sig{\compileK{\pmvB}}{}
                \end{array}$} \\
        \txOut{$(\lambda \vec{\varSig} . \versig{\compileK{\ZCB,\setenum{\pmvA,\pmvB}}}{\vec{\varSig}}, 3 \BTC)$\!} \\
        \hline
      \end{tabular}
      \qquad
      % 
      %%%%%%%%%%%%%%%%%%%%%%%%%%%%%%%%%%%%%%%%%%%%%%%%%%%%%%%%%%%% 
      % TB
      %%%%%%%%%%%%%%%%%%%%%%%%%%%%%%%%%%%%%%%%%%%%%%%%%%%%%%%%%%%% 
      \begin{tabular}{|l|}
        \hline
        \\[-9pt]
        \multicolumn{1}{|c|}{$\txT[\pmvB]$} \\
        \hline
        \\[-9pt]
        \txIn{($\txT[split]$,0)} \\
        \txWit{$\sig{\compileK{\withdrawC{\pmvB},\setenum{\pmvA,\pmvB}}}{}$} \\
        \txOut{$(\lambda \varSig . \versig{\compileK{\pmvB}}{\varSig}, 1 \BTC)$\!} \\
        \hline
      \end{tabular}
      \\[30pt]
      % 
      %%%%%%%%%%%%%%%%%%%%%%%%%%%%%%%%%%%%%%%%%%%%%%%%%%%%%%%%%%%% 
      % Tsplit
      %%%%%%%%%%%%%%%%%%%%%%%%%%%%%%%%%%%%%%%%%%%%%%%%%%%%%%%%%%%% 
      \begin{tabular}{|l|}
        \hline
        \\[-9pt]
        \multicolumn{1}{|c|}{$\txT[split]$} \\
        \hline
        \\[-9pt]
        \txIn{$\txT[init]$} \\
        \txWit{$\sig{\compileK{\ZCB,\setenum{\pmvA,\pmvB}}}{}$} \\
        \txOut{$\begin{array}{l} 
                  0 \mapsto (\lambda \vec{\varSig} . \versig{\compileK{\withdrawC{\pmvB},\setenum{\pmvA,\pmvB}}}{\vec{\varSig}}, 1 \BTC) \\
                  1 \mapsto (\lambda \vec{\varSig} . \versig{\compileK{\afterC{2030}{\withdrawC{\pmvA}},\setenum{\pmvA,\pmvB}}}{\vec{\varSig}}, 2 \BTC)
                \end{array}$\!} \\
        \hline
      \end{tabular}
      \qquad
      % 
      %%%%%%%%%%%%%%%%%%%%%%%%%%%%%%%%%%%%%%%%%%%%%%%%%%%%%%%%%%%% 
      % TA
      %%%%%%%%%%%%%%%%%%%%%%%%%%%%%%%%%%%%%%%%%%%%%%%%%%%%%%%%%%%% 
      \begin{tabular}{|l|}
        \hline
        \\[-9pt]
        \multicolumn{1}{|c|}{$\txT[\pmvA]$} \\
        \hline
        \\[-9pt]
        \txIn{($\txT[split]$,1)} \\
        \txWit{$\sig{\compileK{\afterC{2030}{\withdrawC{\pmvA}},\setenum{\pmvA,\pmvB}}}{}$} \\
        \txOut{$(\lambda \varSig . \versig{\compileK{\pmvA}}{\varSig}, 2 \BTC)$\!} \\
        \txAfterAbs{2030}{} \\
        \hline
      \end{tabular}
    \end{tabular}
  } % resizebox
  \caption{Transactions obtained by compiling the \ZCB contract.}
  \label{fig:zcb:tx}
\end{figure}

The contract $\ZCB$ is compiled into the transactions in~\Cref{fig:zcb:tx}.
The first one that can be appended to the blockchain is $\txT[init]$.
This requires a few conditions to be met:
\begin{inlinelist}
\item $\txT[x_1]$ and $\txT[x_2]$ are \emph{unspent} on the blockchain, 
  \ie no other transactions spend them; 
\item the amount specified in the $\txOut{}$ field of $\txT[init]$
  does not exceed the sum of the amounts in $\txT[x_1]$ and $\txT[x_2]$;
\item the predicates in the $\txOut{}$ fields of $\txT[x_1]$ and 
  $\txT[x_2]$ are true, after replacing the formal parameters with
  the signatures $\sig{\compileK{\pmvA}}{}$ and $\sig{\compileK{\pmvB}}{}$,
  contained in the $\txWit{}$ field of $\txT[init]$.
\end{inlinelist}
The contract $\ZCB$ becomes stipulated once $\txT[init]$ is on the blockchain.

After that, the $\splitname$ action can be performed 
by either $\pmvA$ or $\pmvB$, 
by redeeming $\txT[\it init]$ with $\txT[{\it split}]$. 
This transaction uses
$\compileK{\ZCB,\setenum{\pmvA,\pmvB}}$, a set of two key pairs, 
each one owned by each participant. 
These keys are only used in this step,
to ensure that no transaction but $\txT[\it split]$ can redeem $\txT[init]$.

The transaction $\txT[split]$ creates two unspent outputs
(indexed by $0$ and $1$),
corresponding to the two parallel components of the $\splitname$, 
each with its own balance. 
These outputs can be redeemed independently, by different transactions.
The output at index $0$ % of $\txT[\it split]$ 
can only be redeemed by $\txT[\pmvB]$
(note that $\txT[\pmvB]$'s $\txIn{}$ field refers to 
the output 0 of $\txT[{\it split}]$),
transferring $1\BTC$ to $\pmvB$.  
No other redemption is possible,
since such output requires a signature with a specific key set,
\ie $\compileK{\withdrawC{\pmvB},\setenum{\pmvA,\pmvB}}$, 
which is not used for any other purpose.  
Further, the output of $\txT[\pmvB]$ can
be redeemed with $\pmvB$'s key, without $\pmvA$'s one. 
Similarly, the output $1$ of $\txT[split]$ can be redeemed by
$\txT[\pmvA]$, which in turns transfers $2\BTC$ to $\pmvA$. 
The $\txAfterAbs{}$ field in $\txT[\pmvA]$ 
ensures that this may only happen after time $2030$.

The stipulation protocol followed by participants requires that 
all the signatures needed to append the transactions 
in~\Cref{fig:zcb:tx}
are exchanged \emph{before} $\txT[init]$ is appended.
This is obtained by exchanging the signatures of $\txT[init]$
after all the other signatures.
This ensures that, once the execution of $\ZCB$ starts,
any honest participant can make it proceed, 
by appending a transaction
that correspond to any of the enabled BitML actions.

In general, to guarantee that such liveness property holds, 
the contract must be suitably crafted, using the
$\contrD + \afterC{t}{\contrDi}$
pattern discussed in~\Cref{sec:intro}.
In~\Cref{sec:conclusions} we discuss techniques
to statically verify this property.

\paragraph{Zero-coupon bond with renegotiation.}

Compiling $\ZCBi$ yields the transactions:

\medskip
\begin{nscenter}
\resizebox{1.0\columnwidth}{!}{
  \hspace{-10pt}
  \small
  \begin{tabular}{|l|}
    \hline
    \\[-9pt]
    \multicolumn{1}{|c|}{$\txT[init]$} \\
    \hline
    \\[-9pt]
    \txIn{\hspace{5pt}$\txT_{\txColor{\varX[{1}]}}$} \\
    \txWit{\hspace{1pt}$\sig{\compileK{\pmvA}}{}$} \\
    \txOut{\!$\begin{array}{l}
              (\lambda \vec{\varSig} . \versig{\compileK{\ZCBi,\setenum{\pmvA,\pmvB}}}{\vec{\varSig}}, \\
              \, 1 \BTC)
              \end{array}$\!} \\
    \hline
  \end{tabular}
  \begin{tabular}{|l|}
    \hline
    \\[-9pt]
    \multicolumn{1}{|c|}{$\txT[split]$} \\
    \hline
    \\[-9pt]
    \txIn{$\txT[init]$} \\
    \txWit{$\sig{\compileK{\ZCBi,\setenum{\pmvA,\pmvB}}}{}$} \\
    \txOut{$\begin{array}{l} 
              0 \mapsto (\lambda \vec{\varSig} . \versig{\compileK{\withdrawC{\pmvB},\setenum{\pmvA,\pmvB}}}{\vec{\varSig}}, 1 \BTC) \\
              1 \mapsto (\lambda \vec{\varSig} . \versig{\compileK{\adv{\callX{}},\setenum{\pmvA,\pmvB}}}{\vec{\varSig}},\, 0 \BTC)
            \end{array}$\!} \\
    \hline
  \end{tabular}
  \begin{tabular}{|l|}
    \hline
    \\[-9pt]
    \multicolumn{1}{|c|}{$\txT[\pmvB]$} \\
    \hline
    \\[-9pt]
    \txIn{\hspace{5pt}($\txT[split]$,0)} \\
    \txWit{$\sig{\compileK{\withdrawC{\pmvB},\setenum{\pmvA,\pmvB}}}{}$} \\
    \txOut{$\!\begin{array}{l}
              (\lambda \varSig . \versig{\compileK{\pmvB}}{\varSig}, \\
              \, 1 \BTC)
            \end{array}$\!} \\
    \hline
  \end{tabular}
} % resizebox
\end{nscenter}

\medskip
Once these three transactions are on the blockchain, 
the only enabled action in the corresponding BitML contract 
is $\adv{\callX{}}$, which asks $2\BTC$ from $\pmvB$ as a precondition.
At the Bitcoin level, satisfying this precondition
requires $\pmvB$ to broadcast the identifier of a transaction
$\txT[y]$ holding $2 \BTC$ and redeemable by himself.
In BitML, this corresponds to choosing the deposit name $\varY$
for the deposit variable~$\varphX$.
Then, participants compile the contract advertisement
$\contrAdv{\persdep{\pmvB}{2}{\varphX}}{\contrC}$,
where $\contrC = \afterC{\text{2030}}{\withdrawC{\pmvA}}$,
after replacing $\varphX$ with $\varY$.
The compiler produces the following transactions:

\medskip
\begin{nscenter}
\resizebox{0.7\textwidth}{!}{
  \small
  \begin{tabular}{|l|}
    \hline
    \\[-9pt]
    \multicolumn{1}{|c|}{$\txT[\it init]^{\!\!\!\!\!\!\!\!\!\!\callX{}}$} \\
    \hline
    \\[-9pt]
    \txIn{\hspace{6.5pt}$0 \mapsto (\txT[split],1),1 \mapsto \txT[\varY]$} \\
    \txWit{}:
    \begin{tabular}{l}
    $0 \mapsto\sig{\compileK{\adv{\callX{}},\setenum{\pmvA,\pmvB}}}{}$ \\
    $1 \mapsto \sig{\compileK{\pmvB}}{}$
    \end{tabular}
    \\
    \txOut{\hspace{1pt}$(\lambda \vec{\varSig} . \versig{\compileK{\contrC,\setenum{\pmvA,\pmvB}}}{\vec{\varSig}}, 2 \BTC)$} \\
    \hline
  \end{tabular}
  \qquad
  \begin{tabular}{|l|}
    \hline
    \\[-9pt]
    \multicolumn{1}{|c|}{$\txT[\pmvA]$} \\
    \hline
    \\[-9pt]
    \txIn{$\txT[init]^{\!\!\!\!\!\!\!\!\!\!\callX{}}$} \\
    \txWit{$\sig{\compileK{\contrC,\setenum{\pmvA,\pmvB}}}{}$} \\
    \txOut{$(\lambda \varSig . \versig{\compileK{\pmvA}}{\varSig}, 2 \BTC)$} \\
    \txAfterAbs{2030}\\
    \hline
  \end{tabular}
}
\end{nscenter}

\medskip
The renegotiation succeeds once $\txT[\it init]^{\!\!\!\!\!\!\!\!\!\!\callX{}}$ 
is on the blockchain.
After that, any participant can perform the $\withdrawC{\pmvA}$,
by appending $\txT[\pmvA]$ to the blockchain.

\medskip
As shown by this example, the compiler handles renegotiation as follows:
\begin{itemize}
  
\item at stipulation time, it does not produce transactions 
  for any $\adv{\callX{}}$;
  
\item at renegotiation time,
% when $\adv{}\!$ is executed, 
  the participants broadcast the identifiers of their new deposits, 
  and the commitments of their new secrets. 
  Static expressions are then evaluated, and replaced by their value.
  Finally, the new contract is compiled as usual, with the exception
  that the new initial transaction has an extra input, 
  which transfers the balance of the caller contract to the callee
  (in the $\ZCBi$ example, this extra input is $(\txT[split],1)$ within 
  $\txT[\it init]^{\!\!\!\!\!\!\!\!\!\!\callX{}}\;$).

\end{itemize}

\paragraph{Computational soundness}
The main result of~\cite{BZ18bitml} is computational soundness, which
ensures that each execution trace at the Bitcoin level has a
corresponding one in the semantics of BitML.
This was achieved by formalizing the semantics of Bitcoin using a
computational model, where participants can exchange bitstrings as
messages, and append transactions to the blockchain.
Then, a coherence relation was defined to relate symbolic runs to
computational ones, essentially matching symbolic moves with their
implementation in Bitcoin.

Our extension of BitML with renegotiation still enjoys computational
soundness.
The argument is similar, and requires extending the coherence relation
to the new primitive.
In particular, the reduction: 
\[
\contrAdv[\varX]{\contrG}{\contrC} \mid \confG
\; \xrightarrow{} \;
\contrAdv[\varX]{\contrG}{\contrC} \mid \confG \mid 
\mmid_{i} \,\confSec{\pmvA}{\secrA[i]}{N_i} \mid
\mmid_{j} \,\assign{\pmvA}{\varphX[j]}{\varX[j]} \mid 
\confAuth{\pmvA}{\authCommit[\varX]{\contrG}{\contrC}}
\]
corresponds, in Bitcoin, to $\pmvA$ broadcasting
a message which contains the hashes of her secrets and the transaction 
identifiers that she wishes to use as deposits.

Instead, the reduction:
\[
\contrAdv[\varX]{\contrG}{\contrC} 
 \mid 
\confG 
\; \xrightarrow{} \;
\contrAdv[\varX]{\contrG}{\contrC} 
\mid 
\confG 
\mid
\confAuth{\pmvA}{\authAdv[\varX]{x}{\contrG}{\contrC}}
\]
corresponds to $\pmvA$ signing all the transactions obtained by compiling
the new contract, and broadcasting the signatures.
A participant signs $\txT[\it init]$ only after receiving the
signatures of the other transactions from all the other participants.

Computational soundness requires that each contract involves at least one 
% \emph{honest} 
participant, say $\pmvA$, who follows the Bitcoin implementation of BitML.
In particular, $\pmvA$ follows the stipulation and renegotiation
protocols correctly, \ie signing nothing but the protocol messages,
and signing $\txT[init]$ last.
We also make the usual assumptions on computational adversaries: they
can only run PPTIME algorithms, and they can break the
underlying cryptography with negligible probability, only.
Consequently, we only consider computational runs of polynomial
length (with respect to the security parameter).
This is because in longer runs the adversary would be able to break
the cryptography by brute force.

Below, we provide an intuitive statement of computational soundness.
The formal statement is in~\Cref{app:computational-soundness}.
\begin{theorem}[Computational soundness]
  \label{th:computational-soundness}
  Under the hypotheses above, each Bitcoin-level computational run has
  a corresponding coherent BitML run, with overwhelming probability.
\end{theorem}

%% file: example.tex
\section{A fair recursive coin flipping game}
\label{sec:cfg}

\newcommand{\CFG}{\ensuremath{\contrFmt{\it CFG}}\xspace}
\newcommand{\Split}[1][]{\ensuremath{\contrFmt{{\it Split}_{#1}}}\xspace}
\newcommand{\CFGindent}{\hspace{40pt}}
\newcommand{\CFGindentNP}{\hspace{30pt}}
\newcommand{\CFGtab}{\hspace{10pt}}

To illustrate recursion in our extended BitML,
we introduce a simple game
where two players repeatedly flip coins,
and the one who wins two consecutive flips takes the pot.
The precondition requires each player to deposit $3\BTC$ and choose a secret:
\begin{align*}
  \persdep{\pmvA}{3}{x} \;\mid\; \presecret{\pmvA}{\secrA}
  \; \mid \;
  \persdep{\pmvB}{3}{y} \;\mid\; \presecret{\pmvB}{\secrB}
\end{align*}

\begin{figure}[t]
  \small
  \[
  \begin{array}{l}
    \contrFmt{\it CFG}
    =
    \hspace{6pt}
    \putC[0 \leq \secrB \leq 1]{}{\secrB} . \big( 
    \\
    \CFGindent\CFGtab
    \putC[\secrA = \secrB]{}{\secrA\secrB}. \;
    (
    \adv{\call{\cVarX[\pmvA]}{1}}
    \, + \,
    \afterC{3}{\Split[\pmvA]}
    )
    \\
    \CFGindent
    + \;\putC[\secrA \neq \secrB]{}{\secrA\secrB} . \; 
    (
    \adv{\call{\cVarX[\pmvB]}{1}}
    \, + \,
    \afterC{3}{\Split[\pmvB]}
    )
    \\ 
    \CFGindent
    + \;
    \afterC{2}{\withdrawC{\pmvB}}
    \big)
    \\
    \CFGindentNP
    + \; \afterC{1}{\withdrawC{\pmvA}}
    \\[4pt]
    \call{\cVarX[\pmvA]}{n} 
    = \contrAdv{\presecret{\pmvA}{\secrA} \mid 
    \presecret{\pmvB}{\secrB}}{} 
    \\
    \CFGindent
    \putC[0 \leq \secrB \leq 1]{}{\secrB} . \big( 
    \\
    \CFGindent\CFGtab
    \putC[\secrA = \secrB]{}{\secrA\secrB} . \; \withdrawC{\pmvA}
    \\
    \CFGindent
    + \;\putC[\secrA \neq \secrB]{}{\secrA\secrB} . \;
    (
    \adv {\call{\cVarX[\pmvB]}{n + 1}} 
    \, + \,
    \afterC{(3n+3)}{\Split[\pmvB]}
    )
    \\ 
    \CFGindent + \;
    \afterC{(3n+2)}{\withdrawC{\pmvB}} \big)
    \\
    \CFGindentNP
    + \; \afterC{(3n+1)}{\withdrawC{\pmvA}}
    \\[4pt]
    \call{\cVarX[\pmvB]}{n} 
    = \contrAdv{\presecret{\pmvA}{\secrA} \mid 
    \presecret{\pmvB}{\secrB}}{}
    \\
    \CFGindent
    \putC[0 \leq \secrB \leq 1]{}{\secrB} . \big( 
    \\
    \CFGindent\CFGtab
    \putC[\secrA = \secrB]{}{\secrA\secrB} . \;
    (
    \adv{\call{\cVarX[\pmvA]}{n + 1}}
    \; + \;
    \afterC{(3n+3)}{\Split[\pmvA]}
    )
    \\
    \CFGindent
    + \;\putC[\secrA \neq \secrB]{}{\secrA\secrB} . \; 
    \withdrawC{\pmvB}
    \\ 
    \CFGindent + \;
    \afterC{(3n+2)}{\withdrawC{\pmvB}} \big)
    \\
    \CFGindentNP
    + \; \afterC{(3n+1)}{\withdrawC{\pmvA}}
    \\[4pt]
    \Split[\pmvA] 
    =
    \splitC{(
    \splitB{4}{\withdrawC{\pmvA}} \mid 
    \splitB{2}{\withdrawC{\pmvB}})}
    \\[4pt]
    \Split[\pmvB] 
    =
    \splitC{(
    \splitB{4}{\withdrawC{\pmvB}} \mid 
    \splitB{2}{\withdrawC{\pmvA}})}
  \end{array}
  \]
  \vspace{-10pt}
  \caption{A recursive coin flipping game.}
  \label{fig:cfg}
\end{figure}

The contract \CFG (\Cref{fig:cfg})
asks $\pmvB$ to reveal his secret first:
if $\pmvB$ waits too much, $\pmvA$ can withdraw
the contract funds after time 1.
Then, it is $\pmvA$'s turn to reveal 
(before time 2, otherwise $\pmvB$ can withdraw the funds).
The current flip winner is $\pmvA$
if the secrets of $\pmvA$ and $\pmvB$ are equal, 
otherwise it is $\pmvB$.
At this point, the contract can be renegotiated as 
$\call{\cVarX[\pmvA]}{1}$ or $\call{\cVarX[\pmvB]}{1}$, 
depending on the flip winner
(the parameter $1$ represents the round).
If players do not agree on the renegotiation,
then the funds are split fairly, according to the current expected win.

The contract $\call{\cVarX[\pmvA]}{n}$ requires $\pmvA$ and $\pmvB$
to generate fresh secrets for the $n$-th turn.
If $\pmvA$ wins again, she can withdraw the pot, 
otherwise the contract can be renegotiated as $\call{\cVarX[\pmvB]}{n+1}$.
If the players do not agree on the renegotiation,
the pot is split fairly between them. 
The contract $\cVarX[\pmvB]$ is similar.

The following theorem states that our coin flipping game is fair.
Fairness ensures that the expected payoff of a \emph{rational} player
is always non-negative, notwithstanding the behaviour of the other player.
Rational players must choose random secrets in $\setenum{0,1}$.
Indeed, non uniformly distributed secrets 
can make the adversary bias the coin flip in her favour.
Further, choosing a secret different from $0$ or $1$
would decrease the player payoff.
Indeed, $\pmvB$ would be prevented from revealing his secrets
(by the predicate in the $\putC{}{\secrB}$), 
and so $\pmvA$ could win after the timeout.
If $\pmvA$ chooses a secret different from $0$ or $1$,
she makes $\pmvB$ win the round
(since $\pmvB$ wins when the secrets are different).
Rationality also requires to reveal secrets in time 
(before the alternative $\aftername$ branch is enabled),
% to always agree on renegotiations, 
and to take the $\Split$ branch if restipulation does not occur in time.
This ensures that, when renegotiation happens, 
there is still time to reveal the round secrets.
Indeed, a late renegotiation could enable the other player 
to win by timeout.

\begin{theorem}
  The expected payoff of a rational player
  % who timely reveals secrets  and performs enabled actions 
  is always non-negative.
\end{theorem}

\begin{proof}[Sketch]
  First, we consider the case where renegotiation always happens.  A
  rational player wins each coin flip with probability $1/2$, at
  least: so, the probability of winning the whole game is also $1/2$,
  at least.  In the general case, the renegotiation at the end of each
  round may fail. When this happens, the rational player takes the
  $\Split$ branch, distributing the pot according to the expected
  payoff in the \emph{current} game state, thus ensuring the fairness
  of the whole game. The player who won the last coin flip is expected
  to win $p\BTC$, with
  $p = \nicefrac{1}{2}\cdot 6 +
  \nicefrac{1}{2}\cdot(\nicefrac{1}{2}\cdot p + \nicefrac{1}{2}\cdot
  0)$, giving $p=4$.  Accordingly, the $\Split$ contracts transfer
  $4\BTC$ to the winner of the last flip and $(6-4)\BTC = 2\BTC$ to
  the other player.
\end{proof}

%% file: discussion.tex
\section{More expressive renegotiation primitives}
\label{sec:discussion}

The renegotiation primitive we have proposed for BitML 
is motivated by its simplicity, and by the possibility of compiling 
into standard Bitcoin transactions. 
By adding some degree of complexity, we can devise more general primitives, 
which could be useful in certain scenarios.
We discuss below some alternatives.

\paragraph{Renegotiation-time parameters.}

The primitive $\adv{\callX{\vec{\sexp}}}$ 
allows participants to choose at run-time only the deposit variables 
used in the renegotiated contracts,
and to commit to new secrets.
A possible extension is to allow participants
to choose at run-time \emph{arbitrary} values for the renegotiation 
parameters $\vec{\sexp}$.

For instance, consider a mortgage payment, 
where a buyer $\pmvA$ must pay $10\BTC$ to a bank $\pmvB$ in 10 installments.
After $\pmvA$ has paid the first five installments (of $1\BTC$ each), 
the bank might propose to renegotiate the contract, 
varying the amount of the installment.
Using the BitML renegotiation primitive presented in~\Cref{sec:bitml}, 
we could not model this contract, since the new amount and the number 
of installments are unknown at the time of the original stipulation. 
Technically, the issue is that the primitive $\adv{\callX{\vec{\sexp}}}$
only involves static expressions $\sexp$, 
the value of which is determined at stipulation time. 

% Installment Payment Plan
\newcommand{\IPP}[2][]{\ensuremath{\cVar{IPP}{\langle {#2} \rangle}_{#1}}\xspace}

To cope with non-statically known values, 
we could extend guarded contracts with terms of the form 
$\adv{\callX{\pmvB : \valV}}$, 
declaring that the value $\valV$ is to be chosen by $\pmvB$ 
at renegotiation time.
For instance, this would allow to model our installments payment plan
as $\IPP{1}$, with the following defining equations:
\begin{align*}
  \IPP{\procParamA<5}
  & = \contrAdv{\persdep{\pmvA}{1}{\varphX}}
    {\big(
    \splitC{\splitB{1}{\withdrawC{\pmvB}} \mid \splitB{0}{\adv{\IPP{\procParamA+1}}}}
    \big)}
  \\
  \IPP{5}
  & = \contrAdv{\persdep{\pmvA}{1}{\varphX}}
    {\big(
    \splitC{\splitB{1}{\withdrawC{\pmvB}} \mid \splitB{0}{\adv{\callY{\pmvB:k,\pmvB:v}}}}
    \big)}
  \\
  \callY{\procParamA \neq 1,\procParamB}
  & = \contrAdv{\persdep{\pmvA}{\procParamB}{\varphX}}
    {\big(
    \splitC{\splitB{\procParamB}{\withdrawC{\pmvB}} \mid \splitB{0}{\adv{\callY{\procParamA-1,\procParamB}}}}
    \big)}
  \\
  \callY{1,\procParamB}
  & = \contrAdv{\persdep{\pmvA}{\procParamB}{\varphX}}
    {
    \;\withdrawC{\pmvB}
    }
\end{align*}
where in $\IPP{5}$, the bank chooses the number of installments $k$,
as well as the amount $v$ of each installment.
Note that if $\pmvA$ does not agree with these values,
the renegotiation fails. 
A more refined version of the contract should take this possibility into 
account, by adding suitable compensation branches.
Although adding the new primitive would moderately increase the complexity 
of the semantics and of the compiler, 
this extension can still be implemented on top of standard Bitcoin, 
preserving our computational soundness result.

\paragraph{Renegotiation with a given set of participants.}

\newcommand{\PayOrRefund}{\ensuremath{\contrFmt{\it PayOrRefund}}\xspace}
\newcommand{\Refund}[1][]{\ensuremath{\cVar{Refd}{}_{#1}}\xspace}

As we have remarked in~\Cref{sec:bitml},
a renegotiation can be performed only if \emph{all} 
the participants of the contract agree.
To generalise, we could require the agreement of 
a \emph{given} set of participants 
(possibly, not among those who originally stipulated the contract).

For instance, consider an escrow service between a buyer $\pmvA$ 
and a seller $\pmvB$ for the purchase of an item worth $1\BTC$.
The normal case is that the buyer authorizes the transfer of $1\BTC$ 
after receiving the item, but it may happen that a dishonest seller 
never sends the item, or that a dishonest buyer never authorizes the payment.
To cope with these cases, the participants can renegotiate the contract, 
including an escrow service $\pmvM$ which mediates the dispute,
as follows:
\begin{align*}
  % \Escrow 
  & \authC{\pmvA}{\withdrawC{\pmvB}} \, + \, \authC{\pmvB}{\withdrawC{\pmvA}} 
   + \; \extadv{\pmvA:\pmvM}{\call{\Refund[\pmvA]}{}} 
    \, + \, \extadv{\pmvB:\pmvM}{\call{\Refund[\pmvB]}{}}
  \\[2pt]
  & \Refund[\pmvP]
  = \contrAdv{\persdep{\pmvP}{0.1}{\varphX}}
    {\;\splitC \big(
    \splitB{0.1}{\withdrawC{\pmvM}} \mid \splitB{1}{\withdrawC{\pmvP}}}
    \big)
  % \\
  % \Refund[\pmvB]
  % & = \contrAdv{\persdep{\pmvB}{0.1}{\varphX}}
  %   {\;\splitC \big(
  %   \splitB{0.1}{\withdrawC{\pmvM}} \mid \splitB{1}{\withdrawC{\pmvB}}}
  %   \big)
\end{align*}
where $\extadv{\pmvA:\pmvM}{\call{\Refund[\pmvA]}{}}$
means that only $\pmvA$ and $\pmvM$ need to agree in order 
for the contract $\Refund[\pmvA]$ to be executed, resolving the dispute.
In this case it is crucial that the renegotiation is possible 
even without the agreement between $\pmvA$ and $\pmvB$.
Indeed, if $\pmvM$ decides to refund $\pmvA$
(by authorizing $\Refund[\pmvA]$),
it is not to be expected that also $\pmvB$ agrees.
Similarly to the one discussed before, 
also this extension can be implemented on-top of Bitcoin.
The computational soundness property is preserved, under the assumption that 
at least one participant in any renegotiation is \emph{honest},
\ie it follows the renegotiation protocol.
Crucially, if a renegotiation only involves dishonest participants,
the renegotiated contract could be anything, 
not necessarily that prescribed in the original contract.

\paragraph{Non-consensual renegotiation.}
In the variants of \mbox{$\adv{}$} discussed before, 
renegotiation requires one or more participants to agree.
Hence, each use of \mbox{$\adv{}$} must include 
suitable alternative branches,
to be fired in case the renegotiation fails.
In certain scenarios, we may want to renegotiate the contract 
without the participants having to agree. 
To this purpose, we can introduce a new primitive
$\ncadv{\callX{}}$, which continues as $\callX{}$ without
requiring anyone to agree.
For simplicity, we assume the defining equations
of this primitive of the form
$\decl{\cVarX}{\vec{\procParamA}} = \contrAdv{\valV}{\contrC}$,
where $\valV$ represents the amount of \BTC added to the contract, 
by anyone.

We exemplify the new primitive to design
a two-players game which starts with  
a bet of $1\BTC$ from $\pmvA$, and a bet of $2\BTC$ from $\pmvB$. 
Then, starting from $\pmvA$, players take turns adding $2\BTC$ each to the pot. 
The first one who is not able to provide the additional $2\BTC$ within a given time 
loses the game, allowing the other player to take the whole pot.
The contract is as follows:
\begin{align*}
\contrC 
  & = \setenum{ \persdep{\pmvA}{1}{\varX} \mid \persdep{\pmvB}{2}{\varY}} 
  (\ncadv{\call{\cVarX[\pmvA]}{2}} + \afterC{1}{\withdrawC{\pmvB}})
  \\
  \call{\cVarX[\pmvA]}{n}
  & = \setenum{ 2 } 
    (\ncadv{\call{\cVarX[\pmvB]}{n+1}} + \afterC{n}{\withdrawC{\pmvA}})
  \\
  \call{\cVarX[\pmvB]}{n}
  & = \setenum{ 2 } 
    (\ncadv{\call{\cVarX[\pmvA]}{n+1}} + \afterC{n}{\withdrawC{\pmvB}})
\end{align*}

Unlike \mbox{$\adv{}$}, the action $\ncadv{}$
can be fired without the authorizations of all the players: 
it just requires that the authorization to gather $2 \BTC$ 
is provided, by anyone.
Even though the sender of these $2 \BTC$ is not specified in the contract,
it is implicit in the game mechanism:
for instance, when $\call{\cVarX[\pmvA]}{n}$ calls 
$\call{\cVarX[\pmvB]}{n+1}$, only $\pmvB$ is incentivized to add $2\BTC$,
since not doing so will make $\pmvA$ win.

Implementing the $\ncadv{}$ primitive on top of Bitcoin seems unfeasible: 
even if it were possible to use complex off-chain multiparty computation 
protocols \cite{Gudgeon19iacr}, 
doing so might be impractical. 
Rather, we would like to extend Bitcoin as much as needed 
for the new primitive. 
In our implementation of BitML, 
we compile contracts to sets of transactions and make participants sign them. 
In standard BitML this is doable since, at stipulation time, 
we can finitely over-approximate the reducts of the original contract. 
Recursion can make this set infinite, 
\eg $\call{\cVarX[\pmvA]}{2}, \call{\cVarX[\pmvA]}{3}, \ldots$, 
hence impossible to compile and sign statically. 
A way to cope with this is to extend Bitcoin with \emph{malleable} 
signatures which only cover the part of the transaction not affected 
by the parameter $n$ in $\call{\cVarX[\pmvB]}{n}$. 
Further, signatures must not cover the $\txIn{}$ fields of transactions, 
since they change as recursion unfolds.
In this way, the same signature can be reused for each call. 

Adding malleability provides flexibility, but poses some risks. 
For instance, instead of redeeming the transaction corresponding to 
$\call{\cVarX[\pmvA]}{n}$ with the transaction of 
$\call{\cVarX[\pmvB]}{n+1}$
one could instead use the transaction of 
$\call{\cVarX[\pmvB]}{n+100}$,
since the two transactions have the same signature. 
To overcome this problem, we could add a new opcode 
to allow the output script of $\call{\cVarX[\pmvB]}{n}$ 
to access the parameter in the redeeming transaction, 
so to verify that it is indeed $n+1$ as intended.
Similarly, to check that we have $2\BTC$ more in the new transaction, 
an opcode could provide the value of the new output.
The same goal could be achieved by adapting the
techniques used in~\cite{Moser16bw,Oconnor17bw} to realize \emph{covenants}.

%% file: conclusions.tex
\section{Conclusions}
\label{sec:conclusions}

We have investigated linguistic primitives to renegotiate BitML contracts, and 
their implementation on standard Bitcoin.
More expressive primitives could be devised
by relaxing this constraint, 
\eg assuming the extended UTXO model~\cite{Chakravarty20wtsc}.

The existing verification technique for BitML~\cite{BZ19post}
is based on a sound and complete abstraction of the state space of contracts.
Since this abstraction is finite-state, 
it can be model-checked to verify the required properties.
The same technique can be directly applied to BitML contracts
featuring renegotiation (but without recursion), since the abstraction
would remain finite.
Instead, the same abstraction on recursive contracts would lead to
infinitely many states.
Even if we could exploit the fact that Bitcoin uses 32-bit integers
to make the state space finite, it would still be too large
for verification to be practical.

If we assume that integers are unbounded, 
and that participants always accept renegotiations,
the extension of BitML presented in~\Cref{sec:bitml} 
can simulate a counter machine,
so making BitML Turing-complete.
Hence, verification cannot be sound and complete.
Alternative techniques to model checking 
(\eg, type-based approaches~\cite{Das19arxiv})
could be used to analyse relevant contract properties.

%% file: ack.tex
\paragraph{Acknowledgements}
Massimo Bartoletti is partially supported by Aut.\ Reg.\ of Sardinia projects \textit{Sardcoin} and \textit{Smart collaborative engineering}.
Maurizio Murgia and Roberto Zunino are partially supported by MIUR PON \textit{Distributed Ledgers for Secure Open Communities}.

%% file: semantics-rules.tex
\subsection{Semantics of BitML}

We denote with $\PartT$ the non-empty set of the \emph{honest} participants.
We use notation $\contrAdv[\circ]{\contrG}{\contrC}$ to refer to a contract 
advertisement of the form $\contrAdv{\contrG}{\contrC}$, or
$\contrAdv[\varX]{\contrG}{\contrC}$ where $\varX$ is immaterial.
$\sem{\sexp}$ is the evaluation of static expression $\sexp$, 
defined as expected.
Note that the evaluation function is undefined if $\sexp$ contains free
variables. 
We overload $\sem{-}$ to contracts and contract advertisements:
$\sem{\contrC}$ is the contract obtained by substituting all the occurring 
static expressions with their valuation, and 
$\sem{\contrAdv{\contrG}{\contrC}} = \contrAdv{\contrG}{\sem{\contrC}}$. 
We write $\equiv_{\alpha}$ for equivalence of contract advertisement up-to 
$\alpha$-conversion of secret and deposit names.
%\ie $\contrAdv{\contrG}{\contrC} \equiv_{\alpha} \contrAdv{\contrGi}{\contrCi}$
We write $\contrAdv{\contrG}{\contrC} \equiv \callX{\vec{\sexp}}$ if:
\begin{inlinelist}
\item $\decl{\cVarX}{\vec{\procParamA}} = \contrAdv{\contrGi}{\contrCi}$;
\item $\sem{\contrAdv{\contrGi}{\contrCi}
\subst{\vec{\sem{\sexp}}}{\vec{\procParamA}}} 
\equiv_{\alpha} \contrAdv{\contrG}{\contrC}$.
\end{inlinelist}

The semantic rules for advertisement and stipulation are 
in~\Cref{fig:bitml:semantics:init}.
The rules for actions and for deposits are unchanged (see~\cite{BZ18bitml}).

\begin{figure*}
  \scalebox{.875}{
    \hspace{-20pt}
    \begin{tabular}{c}
      \(
      \irule
      {
      \confG \text{ contains } \confDep[{\varX[i]}]{\pmvA[i]}{\valV[i]} 
      \text{ for all } \persdep{\pmvA[i]}{\valV[i]}{\varX[i]}
      \text{ in } \contrAdv{\contrG}{\contrC}
      \qquad
      \text{$\secretC{\contrG}$ fresh}
      \qquad
      \text{$\partC{\contrAdv{\contrG}{\contrC}} \cap \PartT \neq \emptyset$}
      }
      {\confG 
      \xrightarrow{{\it advertise}(\contrAdv{\contrG}{\contrC})}
      \contrAdv{\contrG}{\contrC} \mid \confG}
      \smallnrule{[C-Adv]}
      \)
      \\[25pt]
      \(
      \irule
      {
      \begin{array}{l}
        \contrAdv{\contrG}{\contrCi} \equiv \callX{\vec{\sexp}}
        \\
        \text{$\secretC{\contrG}$ fresh}
        \\
        \text{$\varphX$ fresh, for each $\persdep{\pmvA}{\valV}{\varphX}$ in $\contrG$}
      \end{array}
      \qquad
      \begin{array}{l}
        \text{all names in } \contrG \text{ fresh}
        \\
        \contrAdv{\contrGi}{\contrCii} \not\equiv \callX{\vec{\sexp}},
        \text{for all } \contrAdv[\varX]{\contrGi}{\contrCii} \in \confG
        \\
        \confG \text{ contains } \confDep[{\varX[i]}]{\pmvA[i]}{\valV[i]}
        \text{ for all } \persdep{\pmvA[i]}{\valV[i]}{\varX[i]}
        \text{ in } \contrG
        %\\
        %\partC{\varX} = \partC{\contrG}
      \end{array}
      }
      {\confContr[\varX]{\adv{\callX{\vec{\sexp}}} + \contrC}{\valV} \mid \confG
      \xrightarrow{{\it advertise}(\contrAdv[\varX]{\contrG}{\contrCi})}
      \confContr[\varX]{\adv{\callX{\vec{\sexp}}} + \contrC}{\valV} \mid \contrAdv[\varX]{\contrG}{\contrCi} \mid \confG}
      \smallnrule{[C-Rngt]}
      \)
      \\[25pt]
      \(
      \irule
      {\begin{array}{c}
         \begin{array}{l}
           \secrA[1] \cdots \secrA[k] \text{ secrets of $\pmvA$ in $\contrG$} 
           \\[3pt]
           \varphX[1] \cdots \varphX[h] \text{ deposit variables of $\pmvA$ in $\contrG$} 
           \\[3pt]
           \forall i \in 1..h : \nexists \varX : (\assign{\pmvA}{\varphX[i]}{\varX}) \in \confG 
           % \\[4pt]
         \end{array}
         \quad
         \begin{array}{l}
           \forall i \in 1..k : \nexists N : \confSec{\pmvA}{\secrA[i]}{N} \in \confG 
           \\[3pt]
           \forall i \in 1..k : N_i \in \begin{cases}
             \Nat & \text{if $\pmvA \in \PartT$} \\
             \Nat \cup \setenum{\bot} & \text{otherwise}
           \end{cases}
           \\
         \end{array}
         \\[5pt]
         \pmvA \in \PartT \implies
         \begin{cases}
           \confG \text{ contains } \confDep[{\varX[i]}]{\pmvA}{\valV[i]}, \text{ for all } i \in 1..h, \text{ and }\\
           {\varX[i]} \neq \varX[j] \text{ for all }i \neq j \in 1..h, 
           \text{ and}
           \\
           {\varX[i]} \neq \varX 
           \text{ for all }i  \in 1..h, \varX \text{ such that }
           \exists \valV : \persdep{\pmvA}{\valV}{\varX}\in\contrG
         \end{cases}
         \\[15pt]
         \confD = \confSec{\pmvA}{\secrA[1]}{N_1} \mid \cdots \mid \confSec{\pmvA}{\secrA[k]}{N_k}\mid\assign{\pmvA}{\varphX[1]}{\varX[1]}\mid\hdots\mid
         \assign{\pmvA}{\varphX[h]}{\varX[h]}
       \end{array}
      }
      {\contrAdv[\circ]{\contrG}{\contrC} \mid \confG
      \xrightarrow{\authLab{\pmvA}{\contrAdv[\circ]{\contrG}{\contrC},\confD}}
      \contrAdv[\circ]{\contrG}{\contrC} \mid \confG \mid \confD
      \mid
      \confAuth{\pmvA}{\authCommit[\circ]{\contrG}{\contrC}}
      }
      \smallnrule{[C-AuthCommit]}
      \)
      \\[25pt]
      \(
      \irule
      {
        % \forall \pmvB : \forall \secrA \text{ secret of $\pmvB$ in $\contrG$} : \exists N : 
        % \big( \confSec{\pmvB}{\secrA}{N} \in \confG \text{ or } \confRev{\pmvB}{\secrA}{N} \in \confG \big)
      \confG \text{ contains } \confAuth{\pmvB}{\authCommit[\circ]{\contrG}{\contrC}}
      \text{ for all $\pmvB$ in $\contrG$} 
      \qquad
      \text{$\contrG = \persdep{\pmvA}{\valV}{\dmv} \mid \cdots$}
      \qquad
      \confG \vdash \dmv = \varX}
      {\contrAdv[\circ]{\contrG}{\contrC} \mid \confG
      \xrightarrow{\authLab{\pmvA}{\contrAdv[\circ]{\contrG}{\contrC},x}}
      \contrAdv[\circ]{\contrG}{\contrC} \mid \confG \mid \confAuth{\pmvA}{\authAdv[\circ]{x}{\contrG}{\contrC}}
      }
      \smallnrule{[C-AuthInitDep]}
      \)
      \\[25pt]
      \(
      \irule
      {
      \confG \text{ contains } \confAuth{\pmvB}{\authCommit[\varX]{\contrG}{\contrC}}
      \text{ for all $\pmvB$ in $\contrG$}}
      {\contrAdv[\varX]{\contrG}{\contrC} \mid \confG
      \xrightarrow{\authLab{\pmvA}{\contrAdv{\contrG}{\contrC},\varX}}
      \contrAdv[\varX]{\contrG}{\contrC} \mid \confG \mid \confAuth{\pmvA}{\authAdv[\varX]{x}{\contrG}{\contrC}}
      }
      \smallnrule{[C-AuthInitContr]}
      \)
      \\[25pt]
      \(
      \irule{
      \begin{array}{c}
        \contrAdv{\contrG}{\contrC} \equiv \callX{\vec{\sexp}}
        \quad
        \varY \text{ fresh}
        \quad
        \contrG = 
        \big( \mmid_{i \in I} \persdep{\pmvA[i]}{\valV[i]}{\varX[i]} \big) \mid
        \big( \mmid_{i \in J} \persdep{\pmvB[i]}{\valVi[i]}{\varphX[i]} \big) \mid
        % \big( \mmid_{i \in K} \tempdep{\pmvC[i]}{\valVii[i]}{\varY[i]} \big) \mid
        % \big( \mmid_{i \in L} \tempdep{\pmvD[i]}{\valViii[i]}{\varphXi[i]} \big) \mid
        \big( \mmid_{i \in K} \presecret{\pmvC[i]}{\secrA[i]} \big)
        \\[3pt]
        \confD = \big( \mmid_{i \in I} \confDep[{\varX[i]}]{\pmvA[i]}{\valV[i]} \big)
        \mid
        \big( \mmid_{i \in J} \confDep[{\varXi[i]}]{\pmvB[i]}{\valVi[i]} \big)
        \mid
        \big( \mmid_{i \in J} \assign{\pmvB[i]}{\varphX[i]}{\varXi[i]} \big)
        % \mid
        % \big( \mmid_{i \in L} \assign{\pmvD[i]}{\varphXi[i]}{\varYi[i]} \big)
        \\[3pt]
        \mid
        \big( \mmid_{\pmvA \in \contrG} \confAuth{\pmvA}{\authCommit[\varX]{\contrG}{\contrC}} \mid \confAuth{\pmvA}{\authAdv[\varX]{x}{\contrG}{\contrC}} \big)
        \mid
        \big( \mmid_{i \in I} \confAuth{\pmvA[i]}{\authAdv[\varX]{\varX[i]}{\contrG}{\contrC}} \big)
        \mid
        \big( \mmid_{i \in J} \confAuth{\pmvB[i]}{\authAdv[\varX]{\varXi[i]}{\contrG}{\contrC}} \big)
      \end{array}
      }
      {\confContr[\varX]{\adv{\callX{\vec{\sexp}}} + \contrCi}{\valV} \mid \contrAdv[\varX]{\contrG}{\contrC} \mid \confG \mid \confD
      \xrightarrow{{\it init}(\varX,\contrG,\contrC)}
      \confContr[\varY]{\contrC}{\sum_{i \in I} \valV[i] + \sum_{i \in J} \valVi[i] + \valV}
      \mid \confG
      }
      \smallnrule{[C-Init]}\)
      \\[25pt]
      \(
      \irule{}{\confG \vdash \varX = \varX} \qquad 
      \irule{}{\confG \mid \assign{\pmvA}{\varphX}{\varX}
      \vdash \varphX = \varX}
      \)
      \qquad
      $cv({\it init}(\varX,\contrG,\contrC)) = \setenum{\varX}$
    \end{tabular}
  }
  \vspace{0pt}
  \caption{Semantics of advertisement and stipulation.}
  \label{fig:bitml:semantics:init}
\end{figure*}

%% file: computational-soundness.tex
\subsection{Compiler}
\label{app:compiler}

\Cref{fig:compiler} shows the new rules to be added to the ones
in~\cite{BZ18bitml} for compiling renegotiations.

Compiling an advertisement used for renegotiation
$\contrAdv[\varZ]{\contrG}{\contrC}$ is similar to compiling a regular
advertisement $\contrAdv{\contrG}{\contrC}$.
The main difference is that, in the renegotiation case, we also have
to redeem the transaction output relative to the parent contract
$\varZ$, denoted with $\txMap{(\varZ)}$.
Therefore, the generated transaction $\txT[\it init]$ includes
$\txMap{(\varZ)}$ as an additional input.
Further, the value $\valV$ of $\txT[\it init]$ consists of the amounts
of the deposits and the balance of the parent contract $\varZ$,
denoted with $\valMap{(\varZ)}$.
The compiler also handles time-constrained renegotiations as in
$\afterC{\constT}{\adv{\callX{\vec{\sexp}}}}$ by using $\constT$ in
the $\txAfterAbs{}$ field of $\txT[\it init]$, so to prevent the
renegotiation to be completed earlier than $\constT$.

Note that compiling $\adv{\callX{\vec{\sexp}}}$ generates \emph{no
  transactions}.
This is because the transactions for $\callX{\vec{\sexp}}$ do not have
to be generated at stipulation time, but only at renegotiation time.

\subsection{Computational soundness}
\label{app:computational-soundness}

Below, $\stratMap(\stratS{\pmvA})$ denotes the computational strategy
obtained by following the symbolic strategy $\stratS{\pmvA}$
of honest participant $\pmvA$, as defined in~\cite{BZ18bitml}.
This is extended to deal with renegotiation using the same protocol as
stipulation, with minor differences as explained
in~\Cref{sec:compiler}.

\begin{theorem}[\textbf{Computational soundness}]
  \label{th:computational-soundness-formal}
  Let $\stratSSet$ be a set of symbolic strategies 
  for all $\pmvA \in \PartT$.
  Let $\stratCSet$ be a set of computational strategies
  such that $\stratC{\pmvA} = \stratMap(\stratS{\pmvA})$
  for all $\pmvA \in \PartT$,
  % for all $\stratC{\pmvA} \in \stratCSet$,
  including an arbitrary adversary strategy $\stratC{\Adv}$.
  Fix $k \in \Nat$.
  Then, the following set has overwhelming probability:
  \[
  \Big\{ 
  \rndR 
  \;\Big|\;
  \begin{array}{l}
    \forall \runC \text{ conforming to $(\stratCSet,\rndR)$ with } |\runC| \leq 
\eta^k : \\
    \qquad
    \exists \runS 
    \text{ conforming to $(\stratSSet,\pi_1(\nonce{}))$ with }
    \coherRel{\runS}{\runC}{\rndR}
  \end{array}
  \Big\}
  \]
\end{theorem}

\begin{figure*}
  \small
  \[
  \begin{array}{c}
    \irule
    {\begin{array}{l}
       \contrG = 
       \big( \mmid_{i \in J} \persdep{\pmvA[i]}{\valVi[i]}{y_i} \big) \mid
       \big( \mmid_{i \in K} \presecret{\pmvB[i]}{a_i} \big)
       \\[4pt]
       \contrC = \sum_{i=1}^{m} \contrD[i]
       \hspace{70pt}
       \valV = \valMap{(\varZ)} + \sum_{i \in J} \valVi[i]
       \\[4pt]
       \expe_i = \compileDscript{\contrD[i]} \;\; (\forall i \in 1..m)
       \qquad
       \vec{x} = \biguplus_{i=1}^{m} \fv{\expe_i}
       \\[4pt]
       \vectxT[i] = \compileD{\contrD[i],\contrD[i],\txT[\it init],0,\valV,\PartGC,0} \;\; (\forall i \in 1..m)
     \end{array}
    \quad
    \scalebox{0.9}{
    \begin{tabular}{|l|}
      \hline
      \\[-8pt]
      \multicolumn{1}{|c|}{$\txT[\it init]$} \\
      \hline
      \\[-6pt]
      \txIn{$\begin{array}{l}
            0 \mapsto \txMap{(\varZ)} \\
            i+1 \mapsto \txMap{(y_i)} \;\; (\forall i \in J)
            \end{array}$%
            } \\[4pt]
      \txWit{$\bot$} \\[2pt]
      \txOut{$\big( \lambda \vec{x} . \bigvee_{i=1}^{m} \expe_i,\, \valV \big)$} \\[2pt]
      \txAfterAbs{$\constT$}{} \\[4pt]
      \hline
    \end{tabular}}
    }
    {\compile{\contrAdv[\varZ]{\contrG}{\contrC}, \constT}
    = \txT[\it init] \vectxT[1] \cdots \vectxT[m]
    }
    \\[20pt]
    \irule
    {\contrD = \adv{\callX{\vec{\sexp}}}
    }
    {\compileD{\contrD,\contrD[p],\txT,o,\valV,\PmvP,\constT} =
    \epsilon \quad \mbox{(empty sequence)}
    }
  \end{array}
  \]
  \caption{Two compilation rules (see~\cite{BZ18bitml} for the others).}
  \label{fig:compiler}
\end{figure*}

%%% Local Variables:
%%% mode: latex
%%% TeX-master: "main"
%%% End: